\documentclass[12pt,onecolumn]{IEEEtran}
\usepackage{amsfonts}
\usepackage{amssymb}
\usepackage{amsmath}
\usepackage{mathrsfs}
\usepackage{verbatim}
\usepackage{subfigure}
\usepackage{balance}
\usepackage{booktabs}
\usepackage{fancybox}
\usepackage{bm}
\usepackage{extarrows}
\usepackage{algorithm}
\usepackage{algorithmic}
\usepackage{multirow}
\usepackage{array}
\usepackage{graphicx}
\usepackage{epstopdf}
\newtheorem{lemma}{Lemma}
\newtheorem{theorem}{Theorem}

\begin{document}
\baselineskip 4.0ex
\title{Impact of Artificial Noise on Cellular Networks: A Stochastic Geometry Approach}

\author{\IEEEauthorblockN{Hui-Ming Wang,~\IEEEmembership{Senior Member, IEEE}, Chao Wang,  Tong-Xing Zheng, ~\IEEEmembership{Student Member, IEEE}, and Tony Q. S. Quek, ~\IEEEmembership{Senior Member, IEEE}}
\thanks{Hui-Ming Wang, Chao Wang, and Tong-Xing Zheng are with the School of Electronic and Information Engineering, and also with
the MOE Key Lab for Intelligent Networks and Network Security,
Xi'an Jiaotong University, Xi'an, 710049, Shaanxi, China. Email:
{\tt xjbswhm@gmail.com}, {\tt \{wangchaoxuzhou,txzheng\}@stu.xjtu.edu.cn}. Their work was partially supported by the Foundation
for the Author of National Excellent Doctoral Dissertation of China under
Grant 201340, the National High-Tech Research and Development Program
of China under Grant No. 2015AA01A708, the New Century Excellent Talents
Support Fund of China under Grant NCET-13-0458, and the Young Talent Support Fund
of Science and Technology of Shaanxi Province under Grant 2015KJXX-01.
}
\thanks{Tony Q. S. Quek is with the Singapore
University of Technology, Singapore. Email: {\tt tonyquek@sutd.edu.sg}}
}


%


\maketitle

\begin{abstract}
This paper studies the impact of  artificial noise (AN) on the secrecy performance of a target cell in  multi-cell cellular networks. Although AN turns out to be an efficient approach for securing a point-to-point/single cell confidential transmission, it would increase the inter-cell interference in a multi-cell cellular network, which may degrade the network reliability and secrecy performance. For analyzing the average secrecy performance of the target cell which is of significant interest, we employ a hybrid cellular deployment model, where
the target cell is a circle of fixed size and the base stations (BSs) outside the target cell are modeled as a homogeneous Poisson point process (PPP). We investigate the impact of AN on the reliability and security of users in the target cell in the presence of pilot contamination using a stochastic geometry approach.
The analytical results of the average connection outage and the secrecy outage of its cellular user (CU) in the target cell are given, which facilitates the evaluation of the average secrecy throughput of a  randomly chosen CU in the target cell. It shows that with an optimized power allocation between the desired signals and AN, the AN scheme is an efficient solution for securing the communications in a  multi-cell cellular network.
\end{abstract}


%
\IEEEpeerreviewmaketitle

\section{Introduction}
Following the pioneering work in \cite{Wyner}, the study on security issue at the physical layer of a communication system has received increasingly attention, especially in wireless communications systems \cite{Enhancingphysicallayer}-\cite{xiaomingchenTSP2014PLS}. In recent years, multiple-input multiple-output (MIMO) technique
has shown to be an effective approach to enhance physical layer security (see \cite{Enhancingphysicallayer} and the references therein.).
Various secrecy signal design schemes have been proposed to increase the \emph{secrecy rate}, which is used to measure the capability of a perfectly secured signal transmission  from an information-theoretic perspective.
In particular, artificial noise (AN) assisted multiple-antenna transmission is a popular secrecy scheme, which was first proposed in \cite{ArtificialNoise}. The basic idea of the AN scheme is to transmit no-information-bearing random signals along with confidential signals to confuse potential eavesdroppers via utilizing extra spatial degrees of freedom provided by multiple antennas. To avoid interfering with the intended legitimate receiver, the AN signal needs to be transmitted in the null space of the legitimate channel \cite{AFF}. Since  the channel state information (CSI) of the eavesdropper is very difficult to be obtained in practice, AN has to be spatial-isotropically broadcasted
such that it can cover all potential eavesdroppers. Without requiring the availability of eavedroppers's CSI,  it has been widely investigated  by a considerable body of literature \cite{zhou_Artificial}-\cite{Zheng2016Optimal} and has also been extended to cooperative relaying networks \cite{Bloch:FriendlyJammingICC10}-\cite{wang2}. For example,  the achievable ergodic secrecy rate and  secrecy throughput were optimized in fast and slow fading multiple-input single-output (MISO) channels, respectively \cite{zhou_Artificial}, \cite{secrecythroughput1}. The AN scheme was also investigated in MIMO channel in \cite{powerallocation} and  a training and feedback based AN scheme was proposed in \cite{WangTSP}. The performance of AN scheme under a randomly distributed eavesdroppers scenario was investigated in  \cite{Zheng2015Multi}, \cite{Zheng2016Optimal}. Cooperative jamming was proposed and optimized in \cite{Bloch:FriendlyJammingICC10}, \cite{ZhengGan}, and was generalized to hybrid jamming schemes in \cite{Wang AF}, \cite{wang2} and uncoordinated jamming schemes in \cite{huimingwangJammerPPP}.

However, in all the above works, the focus was to secure a point-to-point or a single cell wireless transmission, i.e.,  a \emph{single pair} of transmitter and legitimate destination is considered, thus spatial-isotropically AN would not be an issue. Nevertheless, it will not be the case when there are multiple pairs of transmitter and legitimate receiver.
Due to the broadcast nature of wireless medium, spatial-isotropically AN becomes an interference to other concurrent transmissions.
Particularly, for  downlink transmissions in a multi-cell cellular network with universal frequency reuse, inter-cell cochannel interference becomes a critical impairment to the reliability of wireless links. If spatial-isotropically AN is applied  at the BSs to provide secrecy,
additional inter-cell interferences caused by AN will pervade over the cells, which may further deteriorate the network performance. In this case, the application of AN scheme in cellular networks would be questionable.

On the other hand, AN design/optimization requires the prior knowledge of the CSI of the legitimate receivers, which should be obtained via pilot training and channel estimation in practice. For maintaining the bandwidth efficiency, non-orthogonal pilots usually are utilized in different cells. However, this non-orthogonal nature would cause  \emph{pilot contamination} \cite{NoncooperativeUnlimited}, which makes the CSI estimation imperfect. Imperfect legitimate CSI  would result in a so-called \emph{AN leakage} problem, i.e., AN will not be aligned perfectly in the null-space of the legitimate channel so that the intended destination will be disturbed,  which will also bring a significant impact on the performances of the cellular users (CUs).

Therefore, considering the problems mentioned above, the performance of the AN assisted secrecy scheme in a multi-cell cellular network should be evaluated carefully, especially for the impact of AN transmission on the achievable \emph{secrecy} performance of CUs.

\subsection{Related Works}
In the literature, there are several works that studied the physical layer security from a network perspective instead of a point-to-point communication.
A framework of stochastic geometry has been utilized to model the distribution randomness of the users, where both transmitters and receivers are distributed as PPPs.
In \cite{Zhou2011Throughput}, \cite{Enhancing},  single- and multi-antenna secrecy transmissions in an ad hoc network have been investigated.
In \cite{Enhancing}, AN is transmitted along with confidential signal via either sectoring or beamforming. Secrecy outage and secrecy throughput performances are evaluated and optimized. However, perfect CSI has been assumed and pilot contamination problem has not been taken into consideration. Moreover, it is assumed that each transmit-receiver pair over the whole network has a fixed uniform distance, which is not the case in a cellular network.
In \cite{ProtectedZone}, the secrecy performance of AN assisted transmission with a secrecy protected zone has been evaluated in a random network. However, it is a point-to-point communication with only a single pair of secure transceiver.
In \cite{HaowangTWC13PLScellular}, the authors evaluated the achievable secrecy rate of downlink transmissions in cellular networks. They only considered the single-antenna BS case, ignoring both the small-scale fading and inter-cell interference. The work has been extended in \cite{GeraciTCOM2014PLSdownlink}, where  the average secrecy rate of a multiuser downlink transmission via regularized channel
inversion (RCI) precoding was investigated.  A very recent work \cite{xiaomingchenTSP2014PLS} provides a unified secrecy performance analysis to multi-cell MISO downlinks considering the CSI imperfection.
In \cite{SecureMassive}, the authors have extended the investigation to massive MIMO  downlink systems. However, both the analysis in \cite{xiaomingchenTSP2014PLS} and  \cite{SecureMassive} do not take into account the random spatial distribution of BSs and CUs, which is the major deployment manner of the current cellular networks.

\begin{figure}[!t]
  \centering
    \includegraphics[width=0.4\textwidth]{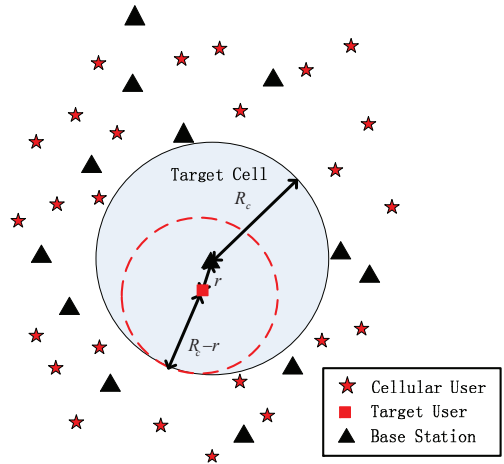}
  \caption{Hybrid model for a multi-cell cellular network. There is a target cell of a fixed size $R_c$. BSs outside the target cell are distributed according to a PPP spatial model. The region inside the dash circle is the interference-exclusive region. An upper bound of the inter-cell interference received at the target user is
the aggregate interference from BSs outside the interference-exclusive region assuming they are PPP distributed over the whole
region outside the dash circle. }
  \label{fig:subfig} 
\end{figure}
\subsection{Main Contributions}
In this paper, taking the AN leakage caused by pilot contamination into consideration, we investigate the impact of additional inter-cell interferences caused by AN transmission on the \emph{reliability} and \emph{secrecy} of a CU under a stochastic geometry framework \cite{spatialmodeling,TractableCellular}. Different from traditional studies under the stochastic geometry framework
where all the BSs over the network are homogeneously distributed and the \emph{network-wide} performances are investigated for the entire system, we take a \emph{cell-specific} perspective where our focus is put on the performance of CUs in a \emph{target cell}. This is because in many applications when cellular networks have been  built out, cellular providers always wonder the performance of some given cells by adding additional BSs in the network.

Motivated by this, in the paper,
we consider a hybrid (stochastic) cellular deployment model, where a target cell we are interested in has a fixed and known shape  and size, and the positions of BSs outside the target cell, CUs, and eavesdroppers are all modeled as independent Poisson point processes (PPPs)\footnote{
The cell-specific perspective has been also proposed in \cite{modelingInterference}, and a hybrid stochastic model is also adopted in \cite{modelingInterferencePPPSP,DASSpectralEfficiency}. In this work, our model is slightly different from \cite{modelingInterference} but have a similar idea.}. Specially, we assume here that the target cell is a circle with a fixed radius $R_c$. Fig. \ref{fig:subfig} depicts the cellular network deployment.
For protecting the confidential information from wiretapping, each BS over the whole network transmits confidential information and AN simultaneously. We analyze the effect of AN transmission in such a network on the connection outage, secrecy outage and average secrecy throughput of a  CU in the target cell. Our goal is to evaluate whether the AN scheme is still valid in a cellular network.
The main contributions of the paper are summarized as follows.
 \begin{enumerate}
 \item Considering the effect of pilot contamination, we provide connection outage and secrecy outage analysis of a CU in the target cell affected by the AN assisted secure transmission scheme in the  random cellular network.
\item  We analyze the achievable average secrecy throughput of a  CU in the target cell by considering the random distribution of CUs and user scheduling.
\item We show that AN is still a promising solution for enhancing the secrecy of users in cellular networks. For maximizing the secrecy performance, the power allocation between the confidential information and AN should be optimized carefully to tradeoff between reliability and secrecy.
 \end{enumerate}

We note that a relative analysis has been provided in \cite{Enhancing} for an ad hoc network. However,
compared with the work in \cite{Enhancing}, the important differences are summarized as follows.
\begin{enumerate}
	\item The analysis models in our work and \cite{Enhancing} are totally different. We adopt the hybrid stochastic model where all BSs and CUs outside the target cell are randomly distributed. But in \cite{Enhancing}, a bipolar network model has been adopted,
	where every transmitter-receiver pair has a fixed distance.
	
	\item We have considered the effect of the pilot contamination and the resulted CSI imperfection and AN leakage problems, while perfect CSI has been assumed in \cite{Enhancing}.
	
	\item We concentrate on analyzing the secrecy performance of CUs in some specific cell, but the work in \cite{Enhancing} analyzes the average secrecy performance of the whole network.
\end{enumerate}

\subsection{Organization and Notations}

\emph{Notation:} $(.)^H$ and $||.||_F$ denote the conjugate transpose and Frobenius norm. $\mathbf{I}_N$ denotes $N\times N$ identity matrix.
$\mathbf{x} \sim \mathcal{CN}(\mathbf{\Lambda}, \mathbf{\Delta})$ denotes the circular symmetric complex Gaussian vector with mean vector ${\mathbf{\Lambda}}$ and variance $\mathbf{\Delta}$, $y \sim \textrm{Gamma}\left(k,\theta\right)$ denotes that $y$ that is gamma-distributed with shape $k$ and scale $\theta$, $y\sim\mathrm{exp}(b)$ denotes that $y$ is an exponential variate whose mean is $b$.
The factorial of a non-negative integer $n$, denoted by $n!$ and $\binom{N}{k} = \frac{N!}{k!(N-k)!}$, $\mathbb{E}$ is the mathematical expectation, ${_2}F_1(\alpha,\beta;\gamma;z)$ denotes the Gauss hypergeometric function \cite[eq. (9.10)]{Table}, and $\gamma(a,x)$ denotes the lower incomplete gamma function \cite[eq. (8.35.1)]{Table}. $\Gamma(x)$ denotes the gamma function \cite[eq. (8.31)]{Table}, and $\Gamma(a,x)$ denotes the  upper incomplete gamma function \cite[eq. (8.35.2)]{Table}.
$b(x,R_c)$ denotes a circle with radius $R_c$ centered at $x$.

\section{System Model and Assumptions}
We consider the secure transmission in a downlink multi-cell cellular network working in  TDD mode with universal frequency reuse, where there are multiple BSs each with $N_t$ antennas, multiple single-antenna CUs, and multiple single-antenna  eavesdroppers.
The downlink transmissions to active  CUs would be wiretapped by potential  eavesdroppers which do not collude. To serve a  CU, the BS transmits Gaussian distributed AN concurrently  with the confidential information. In particular, AN is transmitted spatial-uniformly in the null-space of the estimated legitimate channel from the BS to the  CU \cite{ArtificialNoise}.
Obviously, the AN transmitted from one  BS would be an additional interference for CUs in other cells, especially for its neighbours.

\subsection{Cellular deployment}
As mentioned above, we take a cell-specific perspective and concentrate on analyzing the average secrecy performance of a target cell, whose shape is fixed and known.
In particular, we adopt a so-called hybrid stochastic model following a stochastic geometry framework to model the deployment of the cellular network, which is depicted in Fig. \ref{fig:subfig}. In this model, a \emph{target  cell} of fixed size  is modeled as a circle with radius $R_{\textrm{c}}$ centered at the origin, which is the location of
the target serving BS. The locations of other BSs in the network outside the target cell are modeled as a PPP $\Phi_{\textrm{B}}$ with density $\lambda_{\textrm{B}}$. For the CUs in the target cell, the interferences from other BSs form a shot-noise process \cite{shotnoise}. The shape of interfering cells is determined by the association policy. Here,
the CUs outside the target cell is served by the nearest BS outside the target cell, which implies that the interfering cell area forms a Voronoi tessellation \cite{voronoi}.

\emph{Remark}: Such a model for a multi-cell cellular network was inspired by \cite{modelingInterference}. This model applies to the scenario where the performance achieved in some given region  is of significant interest to the cellular designers.
The reasonability and accuracy of the hybrid model has been well addressed in \cite{modelingInterference}. Using a similar hybrid model, the downlink spectral efficiency of the distributed antenna system has been analyzed in \cite{DASSpectralEfficiency}.

The CUs outside the target cell are distributed as an independent PPP denoted as $\Phi_{\textrm{U}}$, whose intensity is $\lambda_{\textrm{U}}$, respectively.
Time division multiple access (TDMA) is adopted as the multiple access scheme, such that the intra-cell interference is eliminated completely but the inter-cell interference dominates the network performance.
Finally, we model the positions of potential eavesdroppers  as an independent PPP $\Phi_{\textrm{E}}$ with intensity $\lambda_{\textrm{E}}$.

With the above network deployment and association policy, there might be some BSs that do not have any CU to serve, i.e., no CU locates in their Voronoi cells,
and such BSs will not transmit any signal (i.e., inactive). Therefore,
a BS is active if and only if at least one active CU lies in its Voronoi cell.
We should first characterize the distribution of \emph{active} BSs.
According to \cite{voronoi}, the probability density function of the normalized size of a target Voronoi cell can be approximated as  $f_X(x)=\frac{3.5^{3.5}}{\Gamma(3.5)}x^{2.5}e^{-3.5x}$, where $X$ is a random variable that denotes the size of the Voronoi cell normalized by $1/\lambda_{\textrm{B}}$. It is not difficult to get the expectation of $X$ as $\mathbb{E}(X) = 1$.
Therefore, we can approximate the average area of the Voronoi cell as $1/\lambda_{\textrm{B}}$. Then, the spatial distribution of the active BSs outside the target cell can be approximated by an independent thinning of $\Phi_{\textrm{B}}$ with  probability of the non-zero-user event in the cell\footnote{Although the process of the active BSs is an dependent thinning of the initial BS process $\Phi_{\textrm{B}}$. But, for mathematical tractability, just as \cite{DownlinkBSintensity,D2DEnhanced}, we assume that it is an independent thinning of the initial BS process with the thinning probability (in an average sense). The approximation accuracy has been validated by the simulation results given in \cite{DownlinkBSintensity}.}, i.e., a PPP $\hat\Phi_{\textrm{B}}$ with intensity $\hat{\lambda}_{{\textrm{B}}}\approx\left(1-\mathrm{exp}\left(-\frac{\lambda_{\textrm{U}}}{\lambda_{\textrm{B}}}\right)\right)\lambda_{\textrm{B}}$. On the other hand, due to the one-to-one association between active BSs and CUs, the \emph{active CUs} outside the target cell can be approximated as  a PPP $\hat{\Phi}_{\textrm{U}}$ with intensity $\hat{\lambda}_{\textrm{B}}$ as well.
Our goal is to analyze the impact of AN transmission in such a network on the secrecy performances of a CU in the target cell.

\subsection{Channel model}
We consider both large- and small-scale fading for the wireless channels. For large-scale fading, we adopt the standard path loss model $l(r)=d^{-\alpha}$, where $d$ denotes the distance and $\alpha>2$ is the fading exponent \cite{spatialmodeling}. For  small-scale fading, we assume independent quasi-static Rayleigh fading. Since the eavesdroppers are passive wiretappers, their instantaneous CSI and locations are unavailable. Nevertheless, we assume that their small-scale channel distributions are available, which are Rayleigh fading with unit variance.

\subsection{Pilot contamination}
In TDD, with channel reciprocity, the uplink training can provide the BSs with uplink as well as downlink channel estimates \cite{NoncooperativeUnlimited}.
However, a new problem emerges, i.e., ``pilot contamination''. Since  non-orthogonal pilots should be utilized among the cells with  universal frequency reuse, the inter-cell interference causes pilot contamination, which would result in an imperfect CSI estimation.

Under our hybrid model, we now characterize the impact of  pilot contamination on the CSI estimation for CUs in the target cell. A communication begins with the training phase when all the CUs in each cell transmit pilot sequences to their serving BSs.
Without loss of generality,  we assume that a target CU  locates at a distance 
$r$ from the target BS.
Let $\sqrt{\tau}\mathbf{a}\in\mathbb{C}^{\tau\times 1}$ denote the pilot sequence of length $\tau$ transmitted by the CU in each cell during the training phase, where $\mathbf{a}^H\mathbf{a}=1$. The training signal received at the BS in the target cell, $\mathbf{Y}_{\textrm{pilot}}\in \mathbb{C}^{\tau\times N_t}$, is given by:
\begin{align}
\mathbf{Y}_{\textrm{pilot}}=\sqrt{P_{\tau}\tau}r^{-\frac{\alpha}{2}}\mathbf{a}\mathbf{h}_o^T
+\sum_{x\in\hat{\Phi}_{\textrm{U}}}\sqrt{P_{\tau}\tau}\mathbf{a}\mathbf{g}_{x}^Td_x^{-\frac{\alpha}{2}}+\mathbf{N}_{\tau},
\end{align}
where $P_{\tau}$ is the pilot power, $\mathbf{h}_{o}\sim\mathcal{CN}\left(\mathbf{0},\mathbf{I}_{N_t}\right)$ is the small-scale channel vector from the target BS at origin to its served CU,
$\mathbf{g}_{x}\sim\mathcal{CN}\left(\mathbf{0}_{N_t},\mathbf{I}_{N_t}\right)$ is the small-scale  channel vector from  the CU at $x$ to the target BS with distance $d_{x}$ away, and $\mathbf{N}_{\tau}$ is a Gaussian noise matrix having zero mean and variance $N_0$ elements.

Assuming MMSE channel estimation \cite{Hassibi_How_Much_training}, the estimate of $\mathbf{h}_o$ given $\mathbf{Y}_{\textrm{pilot}}$ is obtained as follows:
\begin{align}
\hat{\mathbf{h}}^T_{o}&=\sqrt{P_{\tau}\tau r^{-\alpha}}
\mathbf{a}^H\left(N_0\mathbf{I}_{N_t}+P_{\tau}\tau\mathbf{a}\mathbb{E}_{\mathbf{h}_o,\mathbf{g}_{x},\hat{\Phi}_{\textrm{U}}}\left(\mathbf{h}_o\mathbf{h}_o^H r^{-\alpha}+\sum_{x\in \hat{\Phi}_{\textrm{U}}/b(o,R_c)}\mathbf{g}_{x}\mathbf{g}_{x}^Hd^{-\alpha}_{x}\right)\mathbf{a}^H\right)^{-1}\mathbf{Y}_{\textrm{pilot}}
\nonumber\\
&=\frac{\sqrt{P_{\tau}\tau r^{-\alpha}}}{N_0+P_{\tau}\tau r^{-{\alpha}}+P_{\tau}\tau \mathbb{E}_{\hat{\Phi}_{\textrm{U}}}\left(\sum_{x\in \hat{\Phi}_{\textrm{U}}/b(o,R_c)}d^{-\alpha}_{x}\right)}\mathbf{a}^H\mathbf{Y}_{\textrm{pilot}}.
\end{align}
According to Compbell's Theorem \cite{RandomGraphs}, we have
\begin{align}
\mathbb{E}_{\hat{\Phi}_{\textrm{U}}}\left(\sum_{x\in \hat{\Phi}_{\textrm{U}}/b(o,R_c)}d^{-\alpha}_{U_x}\right)=\hat{\lambda}_{\textrm{B}}\int^{+\infty}_{R_c}\frac{1}{y^{\alpha}}dy=\frac{\hat{\lambda}_{\textrm{B}} R_c^{1-\alpha}}{\alpha-1}.
\end{align}

By the property of MMSE estimation \cite{Hassibi_How_Much_training}, we can express the channel as $\mathbf{h}_o=\hat{\mathbf{h}}_o+\mathbf{e}$, where the estimate $\hat{\mathbf{h}}_o$ and the estimation error $\mathbf{e}\in\mathbb{C}^{1\times N_t}$ are mutually independent with $\hat{\mathbf{h}}_o\sim \mathcal{CN}\left(\mathbf{0}_{N_t}^T,\delta^2\mathbf{I}_{N_t}\right)$ and $\mathbf{e}\sim \mathcal{CN}\left(\mathbf{0}_{N_t}^T,(1-\delta^2)\mathbf{I}_{N_t}\right)$, where \begin{align}
\delta^2\triangleq \frac{P_{\tau}\tau r^{-\alpha}}{N_0+P_{\tau}\tau\frac{\hat{\lambda}_{\textrm{B}} R_c^{1-\alpha}}{\alpha-1}+ P_{\tau}\tau r^{-\alpha}}.\label{estimationerror}
\end{align}

\subsection{Performance metric}
To evaluate the achievable performances of the CUs in the target cell,  we adopt outage constrained performance metrics, which are applicable for delay-sensitive applications, such
as those involving voice or video data communications.

For fighting against eavesdropping, each BS adopts Wyner coding \cite{Wyner} to encode the confidential information. We denote the confidential message rate as $R_s$ and the rate of the transmitted codeword as $R_{t,s}$. Then, the rate redundancy $R_e=R_{t,s}-R_s$ reflects the cost for protecting  the confidential message from wiretapping  \cite{Wyner,secrecythroughput1,secrecythroughput2}. If the channel capacity $C_B$ from the BS to its intended CU is below  $R_{t,s}$, a connection outage event occurs, and the  event probability is defined as \emph{connection outage probability} $p_{co,s}$, which is given by
\begin{align}
p_{co,s}\triangleq\textrm{Pr}\left\{C_{B}\leq R_{t,s}\right\}. \label{defconnoutage}
\end{align}
Accordingly, wiretapped by $k$ non-colluding eavesdroppers, if the maximal channel capacity from the BS
to $k$ eavesdroppers is above the rate $R_e$, a secrecy outage event occur, and the event probability is defined as \emph{secrecy outage probability} $p_{so}$, which is given by
\begin{align}
p_{so}\triangleq\textrm{Pr}\left\{\max_k C_{E_k} > R_e\right\}, \label{defsecoutage}
\end{align}
where $C_{E_k}$ denotes the channel capacity of the $k$th eavesdropper in $\Phi_{\textrm{E}}$.
Under a given connection outage constraint $\sigma$ and secrecy outage constraint $\epsilon$, the \emph{secrecy throughput}
$\mu$ is defined as
\begin{align}
\mu\triangleq\left(1-\sigma\right)R_s.\label{secrecyThroughput}
\end{align}

\section{Connection and Secrecy Outage Analysis}
In this section, we provide connection outage and secrecy outage analysis of a CU in the target cell affected by the AN assisted secure transmission scheme. Here, we first take a \emph{location-specific} perspective, where we focus on one CU in the target cell with a distance $r$ from the BS.
In Section IV, we will go a step further, and analyze the achievable average secrecy throughput of a  CU  in the target cell by considering the random distribution of CUs and user scheduling.

Assume that the total transmit power of each BS is $P_{\textrm{tot}}$. The signal vector $\mathbf{x}_z$ transmitted by a BS at location $z$ is in the form of
\begin{align}
\mathbf{x}_z=\sqrt{P_{\textrm{S}}}\mathbf{w}_zs_z+\sqrt{\frac{P_{\textrm{A}}}{N_t-1}}\mathbf{U}_z\mathbf{n}_{a_z},\label{signalvector}
\end{align}
where $P_{\textrm{S}}=\phi P_{\textrm{tot}}$ is the power to transmit the confidential signal and $P_{\textrm{A}}=\left(1-\phi\right) P_{\textrm{tot}}$  is utilized to transmit AN with $\phi$ the power split factor,
$\mathbf{w}_z=\frac{\hat{\mathbf{h}}_{z}}{||\hat{\mathbf{h}}_{z}||_F}$ is the maximum ratio transmission (MRT) precoding vector with $\hat{\mathbf{h}}_{z}$ the estimate of channel ${\mathbf{h}}_{z}$ from BS at location $z$ to the target CU,
and $\mathbf{U}_z$ is a projection matrix onto the null-space of $\hat{\mathbf{h}}_{z}$, i.e., $\hat{\mathbf{h}}_z^H\mathbf{U}_z=\mathbf{0}$. We note that the columns of $\left[\mathbf{w}_z,\mathbf{U}_z\right]$ constitute an orthogonal basis.

With the transmission strategy described above,  the received signal $y_{\textrm{U}}$ of the  target CU is given by
\begin{align}
y_{\textrm{U}}= \sqrt{P_{\textrm{S}}} r^{-\frac{\alpha}{2}} ||\hat{\mathbf{h}}_{o}||_F s_o
+\underset{\textrm{CSI estimation error}}{\underbrace{\sqrt{P_{\textrm{S}}} r^{-\frac{\alpha}{2}} \mathbf{e}^H\mathbf{w}_os_o}}+\underset{\textrm{AN Leakage}}{\underbrace{\sqrt{\frac{P_{\textrm{A}}}{N_t-1}} r^{-\frac{\alpha}{2}} \mathbf{e}^H\mathbf{U}_o\mathbf{n}_{a_o}}}
+\mathbf{Z}+n_u, \label{ReceivedSignalatCU}
\end{align}
where $\mathbf{Z}$ denotes the aggregated interference from BSs outside the target cell as follows
\begin{align}
\mathbf{Z}&={\underbrace{\sum_{z\in {\hat{\Phi}}_{\textrm{B}}/b(0,R_c)}\left(\sqrt{P_{\textrm{S}}}\mathbf{f}_{z}^H\mathbf{w}_{z}s_{z}+\sqrt{\frac{P_{\textrm{A}}}{N_t-1}} \mathbf{f}_{z}^H\mathbf{U}_{z}\mathbf{n}_{a_{z}}\right)D_{{z}}^{-\frac{\alpha}{2}}}}
,
\nonumber
\end{align}
with $\mathbf{f}_{z}\sim\mathcal{CN}\left(\mathbf{0},\mathbf{I}_{N_t}\right)$ the channel vector and
$D_{z}$ the distance from the interfering BS at $z$ to the target CU, respectively, and $n_u\sim\mathcal{CN}(0,N_0)$ is the Gaussian noise at the CU.

Similarly, the received signals $y_{e}$ of the eavesdropper at $e\in \Phi_{\textrm{E}}$ is given by
\begin{align}
{y}_{e} &= \sqrt{P_{\textrm{S}}} D_{eo}^{-\frac{\alpha}{2}} \mathbf{g}_{eo}^H\mathbf{w}_os_o+\sqrt{\frac{P_{\textrm{A}}}{N_t-1}} D_{eo}^{-\frac{\alpha}{2}}\mathbf{g}_{eo}^H\mathbf{U}_o\mathbf{n}_{a_o}+{n}_e
\nonumber\\
&\qquad +{{\sum_{z\in {\hat{\Phi}}_{\textrm{B}}/b(0,R_c)}\left(\sqrt{P_{\textrm{S}}}\mathbf{g}_{ez}^H\mathbf{w}_{z}s_{z}+\sqrt{\frac{P_{\textrm{A}}}{N_t-1}} \mathbf{g}_{e{z}}^H\mathbf{U}_{z}\mathbf{n}_{a_{z}}\right)D_{e{z}}^{-\frac{\alpha}{2}}}},
\label{ReceivedSignalatEavesdropper}
\end{align}
where
$\mathbf{g}_{ez}\sim\mathcal{CN}\left(\mathbf{0},\mathbf{I}_{N_t}\right)$ and $D_{ez}$ are the channel vector and distance between the  BS at $z$ and the eavesdropper at $e$, respectively, and $n_e\sim\mathcal{CN}(0,\sigma^2_{\textrm{E}})$ is the Gaussian noise at the eavesdropper.

\subsection{Connection Outage Analysis}
We assume that the target CU can obtain the effective channel $||\hat{\mathbf{h}}_{o}||$ for detection via dedicated training \cite{Hassibi_How_Much_training}. With only  $||\hat{\mathbf{h}}_{o}||$ at the target CU,  we consider the worst case where both the CSI estimate error and the AN leakage are modeled as  independent Gaussian noise \cite{Hassibi_How_Much_training}. Defining the SINR thresholds for the connection outage of the target CU  as $\beta_{\textrm{B}_{\textrm{s}}}\triangleq 2^{R_{t,s}}-1$, the connection outage probability of the target CU at a distance $r$ from the target BS can be calculated from (\ref{ReceivedSignalatCU}) as follows
\begin{align}
p_{co,s}(r)\triangleq \textrm{Pr}\left(\frac{P_{\textrm{S}}||\hat{\mathbf{h}}_o||_F^2r^{-\alpha}}{\mathbb{E}\left(\left|\mathbf{e}^H\mathbf{w}_o\right|^2P_{\textrm{S}}r^{-\alpha}+||\mathbf{e}^H\mathbf{U}_o||_F^2\frac{P_{\textrm{A}}}{N_t-1}r^{-\alpha}\right)+I_{out}+N_0}\leq \beta_{\textrm{B}_{\textrm{s}}}\right),\label{connectionoutage}
\end{align}
where the aggregated interference power from the outside BSs is given by
\begin{align}
I_{out}\triangleq \sum_{{z}\in {\hat{\Phi}}_{\textrm{B}}/b(0,R_c)}\left({P_{\textrm{S}}}|\mathbf{f}_{z}^H\mathbf{w}_{z}|^2+{\frac{P_{\textrm{A}}}{N_t-1}} ||\mathbf{f}_{z}^H\mathbf{U}_{z}||_F^2\right)D_{z}^{-{\alpha}}.\nonumber
\end{align}

Before giving the theoretical results of $p_{co,s}(r)$, we first introduce the following lemma.
Defining $P_{z}\triangleq P_{\textrm{S}}|\mathbf{f}_{z}^H\mathbf{w}_{z}|^2 +\frac{P_{\textrm{A}}}{N_t-1} ||\mathbf{f}_{z}^H\mathbf{U}_{z}||_F^2$, we have:
\begin{lemma}
If $P_{\textrm{S}}=\frac{P_{\textrm{A}}}{N_t-1}$, $P_{z}$ is Gamma distributed with shape parameter $N_t$ and scale parameter $P_{\textrm{S}}$, i.e.,
$P_{z}\sim\textrm{Gamma}(N_t,P_{\textrm{S}})$.

If $P_{\textrm{S}}\neq\frac{P_{\textrm{A}}}{N_t-1}$, then the probability density function (pdf) of $P_{z}$ is
\begin{align}
f_{P_{z}}(x)=&\frac{\left(1-\frac{P_{\textrm{A}}}{(N_t-1)P_{\textrm{S}}}\right)^{1-N_t}}{P_{\textrm{S}}\Gamma(N_t-1)}\mathrm{exp}\left(-\frac{x}{P_{\textrm{S}}}\right)
\gamma\left(N_t-1,\left(\frac{N_t-1}{P_{\textrm{A}}}-\frac{1}{P_{\textrm{S}}}\right)x\right).\label{PZpdf}
\end{align}
\end{lemma}
\emph{Proof:} The proof is given in Appendix A.

According to the proof of Lemma 1, since $\mathbf{w}_o$ and $\mathbf{U}_o$ are both independent of $\mathbf{e}$, we have
$\left|\mathbf{e}^H\mathbf{w}_o\right|^2\sim\mathrm{exp}(1-\delta^2)$, and $||\mathbf{e}^H\mathbf{U}_o||_F^2\sim\textrm{Gamma}(N_t-1, 1-\delta^2)$ and therefore
\begin{align}
\mathbb{E}\left(\left|\mathbf{e}^H\mathbf{w}_o\right|^2P_{\textrm{S}}r^{-\alpha}+||\mathbf{e}^H\mathbf{U}_o||_F^2\frac{P_{\textrm{A}}}{N_t-1}r^{-\alpha}\right)=\mathbb{E}
\left(\left|\mathbf{e}^H\mathbf{w}_o\right|^2P_{\textrm{tot}}r^{-\alpha}\right)=\left(1-\delta^2\right)P_{\textrm{tot}}r^{-\alpha}.
\end{align}
Then (\ref{connectionoutage})  can be re-written as
\begin{align}
p_{co,s}(r)\triangleq \textrm{Pr}\left(||\hat{\mathbf{h}}_o||_F^2\leq \delta^2 \mu_{s}(P_{\textrm{I}}+I_{out})\right), \label{Connoutage}
\end{align}
where  $\mu_s\triangleq\frac{\beta_{\textrm{B}_{\textrm{s}}}r^{\alpha}}{P_{\textrm{S}}\delta^2}$ and $P_{\textrm{I}}\triangleq\left(1-\delta^2\right)P_{\textrm{tot}}r^{-\alpha}+N_0$.

From (\ref{Connoutage}), the critical step to evaluate $p_{co,s}(r)$ lies in providing a tractable form of $I_{out}$, which is unfortunately difficult due to the  asymmetry of interference region to the target CU.
From Fig. \ref{fig:subfig}, we can find that for the target CU, the interfering region is asymmetric since the distance to the closet edge of the cell is $R_c-r$ while to the furthest edge is $R_c+r$.
The asymmetric property renders the exact result of $p_{co,i}(r)$  difficult to obtain.
To avoid the dependence on the location, we employ the ``small ball'' approximation illustrated by Fig. \ref{fig:subfig} to get a safe approximate  \cite{modelingInterference}.
In particular, we consider a reduced interference-exclusive region, which is a ball of radius $R_u\triangleq R_c-r$ with target CU at the center.
The aggregate interference from BSs outside the interference-exclusive region assuming they are PPP distributed over the whole
region outside the dash circle
is an upper bound of the inter-cell interference received at the target CU.
With such an approximation, a conservative secrecy performance can be obtained, and such an approximate performance analysis method has also been adopted in
 \cite{modelingInterference}.
The approximate theoretical results are given by the following theorem.
\begin{theorem}
Denoting 
$
\hat{I}_{out}\triangleq \sum_{{z}\in {\hat{\Phi}}_{\textrm{B}}/b(0,R_u)}P_{z}D_{z}^{-{\alpha}},\nonumber
$
a safe approximate $p_{co,s}$ is given by
\begin{align}
&p_{co,s}(r)\lessapprox \hat{p}_{co,s}=1-\mathrm{exp}\left(-\mu_s P_{\textrm{I}}\right)\sum_{k=0}^{N_t-1}\sum^{k}_{p=0}\frac{\left(\mu_sP_{\textrm{I}}\right)^{k-p}x_{p,s}}{(k-p)!},\label{ApproximateConnectionOutage}
\end{align}
where $x_{0,s}\triangleq\mathcal{L}_{\hat{I}_{out}}(\mu_s)$, $x_{p,s}\triangleq\sum_{m=1}^{N_t-1}\frac{\mathbf{Q}^m(p+1,1)}{m!}
\mathcal{L}_{\hat{I}_{out}}(\mu_s)$,
\begin{align}
\mathbf{Q}\triangleq&\left[
\begin{matrix}
0,&&\\
\Psi_1,&0&\\
\Psi_2,&\Psi_1,&0\\
\vdots&&&\ddots\\
\Psi_{N_t-1},&\Psi_{N_t-2}&\cdots&\Psi_{1}&0
\end{matrix}
\right],
\nonumber\\
\Psi_m\triangleq&\left\{
\begin{array}{ll}
\mu_s^{m}\left(2\pi\hat{\lambda}_{\textrm{B}}\frac{\left(1-\frac{P_{\textrm{A}}}{(N_t-1)P_{\textrm{S}}}\right)^{1-N_t}}{P_{\textrm{S}}}
\Upsilon_{m}
\right),&\textrm{ if } P_{\textrm{S}}\neq \frac{P_{\textrm{A}}}{N_t-1}
\\
\mu_s^{m}2\pi\hat{\lambda}_{\textrm{B}}
\Theta_{m},&\textrm{ if } P_{\textrm{S}}= \frac{P_{\textrm{A}}}{N_t-1}
\end{array}
\right.
\nonumber
\end{align}
\begin{align}
\Theta_{m}\triangleq&\frac{\binom{N_t+m-1}{N_t-1}P_{\textrm{S}}^{m}\left(R_u^{-\alpha}\right)^{m-\frac{2}{\alpha}}{_2}F_1\left(m+N_t,m-\frac{2}{\alpha};m-\frac{2}{\alpha}+1;-\mu_sP_{\textrm{S}}R_u^{-\alpha}\right)}{\alpha(m-\frac{2}{\alpha})},
\end{align}
\begin{align}
\Upsilon_{m}\triangleq&\frac{P_{\textrm{S}}^{m+1}R_u^{-\alpha\left(m-\frac{2}{\alpha}\right)}}{\left(m-\frac{2}{\alpha}\right)\alpha}
{_2}F_1\left(m+1,m-\frac{2}{\alpha};m-\frac{2}{\alpha}+1;-{\mu_sP_{\textrm{S}}R_u^{-\alpha}}\right)\nonumber\\
&-\sum_{i=0}^{N_t-2}\binom{i+m}{i}\left(\frac{\left(\frac{N_t-1}{P_{\textrm{A}}}-\frac{1}{P_{\textrm{S}}}\right)^i}{\alpha}\right)
\left(\frac{P_{\textrm{A}}}{N_t-1}\right)^{i+m+1}\frac{\left(R_u^{-\alpha}\right)^{m-\frac{2}{\alpha}}}{m-\frac{2}{\alpha}}
\nonumber\\
&{_2}F_1\left(i+m+1,m-\frac{2}{\alpha};m-\frac{2}{\alpha}+1;-\frac{\mu_sP_{\textrm{A}}R_u^{-\alpha}}{N_T-1}\right)
,
\label{Upsilon2m}
\end{align}
and $\mathcal{L}_{\hat{I}_{out}}(\mu_s)$ can be calculated by (\ref{equivalent}) and (\ref{Noequivalent}).

1) If $P_{\textrm{S}}=\frac{P_{\textrm{A}}}{N_t-1}$,
\begin{align}
&\mathcal{L}_{\hat{I}_{out}}(\mu_s)=\mathrm{exp}\left(\hat{\lambda}_{\textrm{B}}\pi\left(R_u^2-\frac{R_u^2}{\left(P_{\textrm{S}}R_u^{-\alpha}\mu_s+1\right)^{N_t}}
-\frac{\mu_s^{\frac{2}{\alpha}}}{P_{\textrm{S}}^{N_t}\Gamma(N_t)}
\frac{\left(R_u^{-\alpha}\mu_s\right)^{1-\frac{2}{\alpha}}\Gamma(N_t+1)}{\left(1-\frac{2}{\alpha}\right)\left(R_u^{-\alpha}\mu_s+\frac{1}{P_{\textrm{S}}}\right)^{N_t+1}}
\right.\right.\nonumber\\
&\left.\left. {_2}F_1\left(1,N_t+1,2-\frac{2}{\alpha},\frac{R_u^{-\alpha}\mu_s}{R_u^{-\alpha}\mu_s+\frac{1}{P_{\textrm{S}}}}\right)\right)\right)
\label{equivalent}
\end{align}

2) If $P_{\textrm{S}}\neq \frac{P_{\textrm{A}}}{N_t-1}$,
\begin{align}
&\mathcal{L}_{\hat{I}_{out}}(\mu_s)=\mathrm{exp}\left(-\hat{\lambda}_{\textrm{B}}\pi\frac{\gamma(\delta+N_t)}{P_{\textrm{S}}\gamma(N_t)\left(\frac{N_t-1}{\textrm{P}_{\textrm{A}}}\right)^{\delta+1}}{_2}F_1\left(1,N_t+\delta;N_t;1-\frac{P_{\textrm{A}}}{(N_t-1)P_{\textrm{S}}}\right)\Gamma(1-\delta)\mu_s^{\delta}+T(\mu_s)\right)
.
\label{Noequivalent}
\end{align}
where
\begin{align}
&T(\mu_s)\triangleq\pi\hat{\lambda}_{\textrm{B}}R_u^2-\frac{2\pi\hat{\lambda}_{\textrm{B}}\left(1-\frac{P_{\textrm{A}}}{(N_t-1)P_{\textrm{S}}}\right)^{1-N_t}}{\alpha P_{\textrm{S}}}\left(\frac{R_u^{2+\alpha}}{\mu_s\left(1+\frac{2}{\alpha}\right)}{_2}F_1\left(1,1+\frac{2}{\alpha};2+\frac{2}{\alpha};-\frac{1}{P_{\textrm{S}}\mu_sR_u^{-\alpha}}\right)
-\right.
\nonumber\\
& \sum^{N_t-2}_{i=0}\frac{\left(1-\frac{P_{\textrm{A}}}{(N_t-1)P_{\textrm{S}}}\right)^i}{\frac{N_t-1}{P_{\textrm{A}}}}
\frac{R_u^{2+\alpha i+\alpha}}{\left(\frac{P_{\textrm{A}}\mu_s}{N_t-1}\right)^{i+1}\left(i+1+\frac{2}{\alpha}\right)}{_2}F_1
\left(i+1,i+1+\frac{2}{\alpha};i+2+\frac{2}{\alpha};-\frac{N_t-1}{P_{\textrm{A}}\mu_sR_u^{-\alpha}}\right)
\label{tmui}.
\end{align}
\end{theorem}
\emph{Proof:}
The proof is given in Appendix B.

The analytical result of the approximation  given in Theorem 1
is rather unwieldy.
With Alzer's inequality \cite{Alzer}, we next provide a tight lower bound of $\hat{p}_{co,i}$, which can simplify the  theoretical calculation of the connection outage.
\begin{theorem}
Denoting $\kappa=(N_t!)^{-\frac{1}{N_t}}$, $\hat{p}_{co,s}(r)$  is tightly lower bounded by
\begin{align}
\hat{p}_{co,s}(r)\geq\hat{p}^L_{co,s}(r)= 1+\sum_{k=1}^{N_t}(-1)^{k}\binom{N_t}{k}\mathrm{exp}\left(-k\kappa\mu_s P_{\textrm{I}}\right)\mathcal{L}_{\hat{I}_{out}}(k\kappa\mu_s).\label{pcolowerBound}
\end{align}
\end{theorem}
\emph{Proof:}
The proof is given in Appendix C.

\begin{figure}[!t]
\centering
\includegraphics[width=3in]{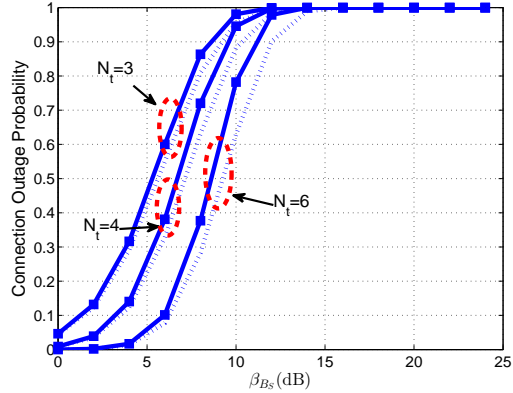}
\caption{Validation of the analysis result $\hat{p}_{co,s}$ and its lower bound $\hat{p}^L_{co,s}$ given in Theorem 2
for $\lambda_{\textrm{B}} = 10^{-4},R_c=200,r=0.25*R_c$,$\lambda_{\textrm{U}}=10\lambda_{\textrm{B}}$,
$\tau=N_t$, $P_{\tau}=20$dBm, $\alpha=3,P_{\textrm{tot}}=30$dBm, $\phi=0.5$, and $N_0=-50$ dBm.}
\label{Comparison}
\end{figure}
The analysis result of the connection outage in Theorem 1 and its lower bound in Theorem 2 are validated in Fig. \ref{Comparison}. For all the simulations in this paper, 100, 000 trials are used. From the simulation results in Fig. \ref{Comparison}, we can observe that the ``small ball'' approximation is sufficiently accurate and the simulation result almost overlaps with the analytical expression (\ref{ApproximateConnectionOutage}) in Theorem 1. Furthermore,
the lower bound is tight, especially for the low connection outage region.

\begin{figure}[!t]
\centering
\includegraphics[width=3in]{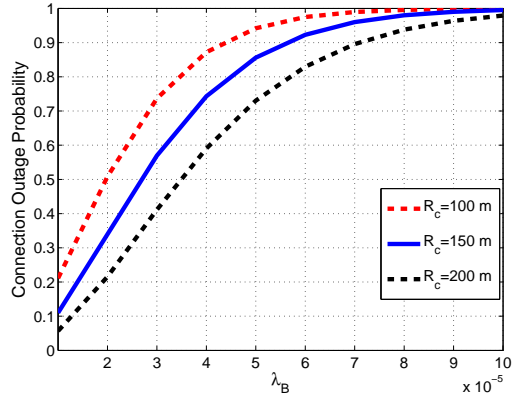}
\caption{$\hat{p}_{co,s}$ versus $\lambda_B$ for $\beta_B=10$ dB, $r=50$ m, $R_c=200$, $\lambda_{\textrm{U}}=10\lambda_{\textrm{B}}$,
$\tau=N_t, N_t=3$, $P_{\tau}=20$dBm, $\alpha=3,P_{\textrm{tot}}=30$dBm, $\phi=0.5$, and $N_0=-70$ dBm.}
\label{ChangWithLambdaB}
\end{figure}

In the following, we analyze the effect of the key system parameters on the connection outage probability. From the  model as illustrated in Fig. \ref{fig:subfig}, we know that the out-cell interference increases with the increasing $\lambda_B$ and the decreasing $R_c$. Therefore, we can infer that the CU's connection outage probability would increases with the increasing $\lambda_B$ and decreasing $R_c$.
The effect of $\lambda_{\mathrm{B}}$ on the CU's connection outage probability can also be found from the analysis result given in (\ref{ApproximateConnectionOutage}).
Assuming taht $P_I=0$, $\hat{p}_{co,s}$ would degrade into
\begin{align}
&\hat{p}_{co,s}=1-\sum^{N_t-1}_{p=0}{x_{p,s}}.\nonumber
\end{align}
Since $\Psi_m$ is an increasing function of $\lambda_{\mathrm{B}}$, we can infer that $\hat{p}_{co,s}$ is a decreasing function of $\lambda_{\mathrm{B}}$.
$P_{\tau}$ and $\tau$ determine the channel estimation quality, i.e., $P_I$ in (\ref{ApproximateConnectionOutage}). Therefore, obviously, the CU's connection outage probability would increases with the decreasing $P_{\tau}$ and $\tau$.
The simulation results in  Fig. \ref{ChangWithLambdaB} show the change trend of the CU's connection outage probability versus $\lambda_{\mathrm{B}}$ for different $R_c$. From the simulation results in Fig. \ref{ChangWithLambdaB}, we can find that the CU's connection outage probability increases with the increasing $\lambda_{\mathrm{B}}$ and the decreasing $R_c$, which has validated the above analysis.

\subsection{Secrecy Outage Analysis}
In this subsection, we characterize the secrecy outage probability $p_{so}$ of a target CU. Since the noise power at each eavesdropper is unknown, just as
\cite{secrecythroughput1,wang2,huimingwangJammerPPP,Enhancing},
considering the worst-case, we set it to be zero, i.e., $\sigma^2_{\textrm{E}}=0$. Furthermore, in order to achieve the maximum level of secrecy,  we assume that eavesdroppers are capable of performing multiuser
decoding (e.g., successive interference cancellation) and the concurrent transmissions of information signals in other cells would not degrade the quality of reception at eavesdroppers \cite{Enhancing}. Then the wiretapping is only disturbed by the AN, and thus the achievable SIR at an eavesdropper at $e\in \Phi_{\textrm{E}}$ can be calculated  as
\begin{align}
\textrm{SIR}_{E_e}=\frac{P_{\textrm{S}}|\mathbf{g}_{eo}^H\mathbf{w}_o|^2D_{eo}^{-\alpha}}{\frac{P_{\textrm{A}}}{N_t-1}
||\mathbf{g}_{eo}^H\mathbf{U}_o||_F^2D_{{eo}}^{-\alpha}+
{\sum_{{z}\in {\hat{\Phi}}_{\textrm{B}}/b(0,R_c)}{\frac{P_{\textrm{A}}}{N_t-1}} ||\mathbf{g}_{e{z}}^H\mathbf{U}_{z}||_F^2}
D_{e{z}}^{-{\alpha}}}.
\end{align}

Assuming that the SIR threshold for secrecy outage is defined as $\beta_{\textrm{E}}=2^{R_e}-1$,  from the definition in (\ref{defsecoutage}), the secrecy outage is given by
\begin{align}
p_{so} = 1- \textrm{Pr}\{\max C_{E_e}\leq R_e\}
=1-\mathbb{E}_{\hat{\Phi}_{\textrm{B}}}\left(\mathbb{E}_{\Phi_{\textrm{E}}}\left(\prod_{e\in\Phi_{\textrm{E}}}
\textrm{Pr}\left(\textrm{SIR}_{E_e}<\beta_{\textrm{E}}|\hat{\Phi}_{\textrm{B}}\right)
\right)\right).\label{pso}
\end{align}
Unfortunately, deriving an accurate analytical expression of (\ref{pso}) is mathematically intractable. Instead, we provide upper and lower bounds of (\ref{pso}) in the following theorem.
\begin{theorem}
Denoting $\alpha_{\textrm{E}}\triangleq \frac{P_{\textrm{A}}\beta_{\textrm{E}}}{(N_t-1)P_{\textrm{S}}}$,
the upper bound $p^U_{so}$ and lower bound $p^L_{so}$ of the secrecy outage probability are given by
\begin{align}
p^U_{so}=&1-\mathrm{exp}\left(-2\pi\lambda_{\textrm{E}}\left(1+\alpha_{\textrm{E}}\right)^{-N_t+1}\left(\int^{R_c}_0
\mathrm{exp}\left(-\lambda_{\textrm{B}_{\textrm{s}}}\left(\int^{\pi}_0\left(\Omega\left(l_2(\theta)\right)+\Omega\left(l_1(\theta)\right)\right)\right)\right)ydy+
\right.\right.\nonumber\\
&\left.\left.\int^{+\infty}_{R_c}\mathrm{exp}\left(-2\lambda_{\textrm{B}_{\textrm{s}}}\left(\int^{\nu}_0\left(\Xi_1(\theta)+\Omega\left(l_4(\theta)\right)\right)d\theta
+\int^{\pi}_\nu \Xi_2(\theta)d\theta\right)\right)
ydy\right)\right),\label{upperBoundSecrecyOutage}
\end{align}
 \begin{align}
p^L_{so}=&\left(1+\alpha_{\textrm{E}}\right)^{-N_t+1}\left(\int^{R_c}_0
2\pi\lambda_{\textrm{E}}ye^{-\pi\lambda_{\textrm{E}} y^2}\mathrm{exp}\left(-\lambda_{\textrm{B}_{\textrm{s}}}\left(\int^{\pi}_0\left(\Omega\left(l_2(\theta)\right)+\Omega\left(l_1(\theta)\right)\right)\right)\right)ydy
\right.\nonumber\\
&\left.+\int^{+\infty}_{R_c}2\pi\lambda_{\textrm{E}}ye^{-\pi\lambda_{\textrm{E}} y^2}\mathrm{exp}\left(-2\lambda_{\textrm{B}_{\textrm{s}}}\left(\int^{\nu}_0\left(\Xi_1(\theta)+\Omega\left(l_4(\theta)\right)\right)d\theta
+\int^{\pi}_\nu \Xi_2(\theta)d\theta\right)\right)ydy\right).\label{LowerBoundSecrecyOutage}
 \end{align}
where  $l_1(\theta)\triangleq\sqrt{R_c^2-y^2\textrm{sin}^2\theta}+y\textrm{cos}\theta$, $l_2(\theta)\triangleq\sqrt{R_c^2-y^2\textrm{sin}^2\theta}-y\textrm{cos}\theta$,
$l_3(\theta)\triangleq y\textrm{cos}\theta-\sqrt{R_c^2-\left(y\textrm{sin}\theta\right)^2}$,
$l_4(\theta)\triangleq l_3(\theta)+2\sqrt{R_c^2-\left(y\textrm{sin}\theta\right)^2}$, $\nu\triangleq \textrm{arsin}\left(\frac{R_c}{y}\right)$, and
\begin{align}
\Omega(x)\triangleq&
-\frac{x}{2}\left(1-\left(1+{\alpha_{\textrm{E}}y^{\alpha}x^{-\alpha}}\right)^{-N_t+1}\right)
+\left({\alpha_{\textrm{E}}y^{\alpha}}\right)^{\frac{2}{\alpha}}\frac{\Gamma\left(1-\frac{2}{\alpha}\right)}{2}\frac{\Gamma\left(N_t+\frac{2}{\alpha}-1\right)}{\Gamma\left(N_t-1\right)}
\nonumber\\
&-\left(\frac{P_{\textrm{A}}\beta_Ey^{\alpha}}{2P_{\textrm{S}}}\right)\frac{x^{2-\alpha}}{\left(N_t+\frac{2}{\alpha}-1\right)\left({\alpha_{\textrm{E}}y^{\alpha}x^{-\alpha}}+1\right)^{N_t}}
{_2}F_1\left(1,N_t;N_t+\frac{2}{\alpha};\frac{1}{{\alpha_{\textrm{E}}y^{\alpha}x^{-\alpha}}+1}\right),
\end{align}
\begin{align}
\Xi_1(\theta)\triangleq &\frac{l_3^2(\theta)}{2}\left(1-\left(1+{\alpha_{\textrm{E}}y^{\alpha}l_3^{-\alpha}(\theta)}\right)^{-N_t+1}\right)
+\frac{l_3^{2-\alpha}(\theta)}{\left(N_t+\frac{2}{\alpha}-1\right)\left({\alpha_{\textrm{E}}y^{\alpha}l_3^{-\alpha}(\theta)}+1\right)^{N_t}}
\nonumber\\
&\left(\frac{P_{\textrm{A}}\beta_{\textrm{E}}y^{\alpha}}{2P_{\textrm{S}}}\right){_2}F_1\left(1,N_t;N_t+\frac{2}{\alpha};\frac{1}{{\alpha_Ey^{\alpha}l_3^{-\alpha}(\theta)}+1}\right)
,\label{Xi1}\\
\Xi_2(\theta)\triangleq&\frac{1}{2}\left({\alpha_{\textrm{E}}y^{\alpha}}\right)^{\frac{2}{\alpha}}\Gamma\left(1-\frac{2}{\alpha}\right)\frac{\Gamma\left(N_t+\frac{2}{\alpha}-1\right)}{\Gamma\left(N_t-1\right)}. \label{Xi3}
\end{align}

\end{theorem}
\begin{proof}
The proof is given in Appendix D.
\end{proof}

Theoretical results in (\ref{upperBoundSecrecyOutage}) and (\ref{LowerBoundSecrecyOutage}) are validated by simulation results in Fig. \ref{SecrecyOUtageTest}, where we can observe that the upper bound is  tight.
Therefore, we will use ${p}_{so}^U$ for approximating $p_{so}$.

\begin{figure}[!t]
\centering
\includegraphics[width=3in]{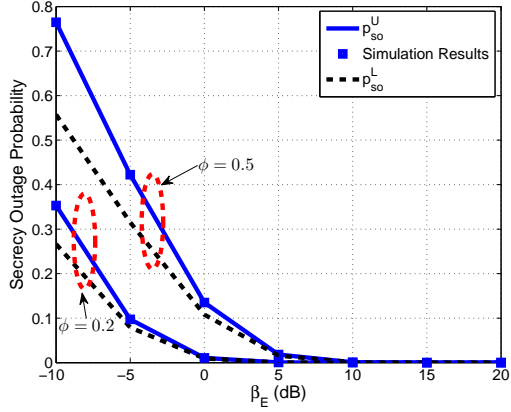}
\caption{Validation of the upper bound $p^U_{so}$ and lower bound $p^L_{so}$ for $\lambda_{\textrm{B}} = 1/(16\times200^2),R_c=300$,$\lambda_{\textrm{U}}=10\lambda_{\textrm{B}}$, $\lambda_{\textrm{E}}=2\lambda_{\textrm{B}}$, $N_t=4$,
$\tau=N_t$, $P_{\tau}=20$dBm, $\alpha=3,P_{\textrm{tot}}=30$dBm,  $N_0=-50$ dBm.}
\label{SecrecyOUtageTest}
\end{figure}

\begin{figure}[!t]
\centering
\includegraphics[width=3in]{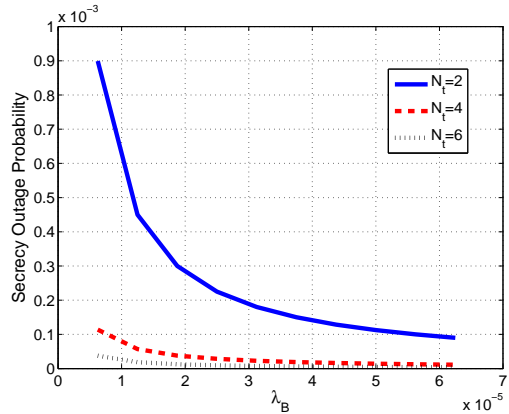}
\caption{Secrecy outage probability versus $\lambda_{\textrm{B}}$ for $R_c=300$, $\lambda_{\textrm{U}}=10\lambda_{\textrm{B}}$, $\lambda_{\textrm{E}}=\frac{1}{16\times300^2}$,
$\tau=N_t$, $P_{\tau}=20$dBm, $\alpha=3,P_{\textrm{tot}}=30$dBm, $\phi=0.5$, $N_0=-50$ dBm.}
\label{SecrecyOutageWithLambdaB}
\end{figure}

In the following, we analyze the effects of the key system parameters on the achievable secrecy outage probability. Firstly, with the increasing $\lambda_B$, the network interference due to AN transmitted from multiple active  BSs would increase, and the wiretapping capability of eavesdroppers would decrease. Therefore, we can infer that the secrecy outage probability would decrease with the increasing $\lambda_B$. Secondly, as the number of antennas equipped at each BS increases, more degrees-of-freedom can be utilized by each BS for transmitting AN. Therefore, we can infer that the secrecy outage probability would decrease with the increasing $N_t$.
Finally, since the cellular network is interference-limited, the network interference dominates the reception quality of eavesdroppers. Therefore, we can infer that the transmit power of each BS has little or no effect on the achievable secrecy outage probability. The simulation results in  Fig. \ref{SecrecyOutageWithLambdaB} show the change trend of the secrecy outage probability versus $\lambda_B$ for different $N_t$. From the simulation results in Fig. \ref{SecrecyOutageWithLambdaB}, we can find that the secrecy outage probability decreases with the increasing $\lambda_B$ and $N_t$, which has validated the above analysis.

It is worth mentioning that although the network interference would increase the connection outage probability of the target user, it also can decrease  the secrecy outage probability. Therefore, the network interference may not be harmful for the cellular communication, when considering the communication security.
\section{Average Secrecy Throughput and Data Throughput Analysis}
The analysis results in the above section is for a target CU at a certain distance $r$ from the target BS. We note that  the connection outage probabilities in Theorem 1 and Theorem 2 are functions of $r$ ($\mu_s$ is a function of $r$), i.e., location dependent, while the secrecy outage probability is independent to $r$. This is because a secrecy outage occurs when eavesdroppers have a better channel than the threshold, which is irrespective to the location of the target CU.
In this section, we will analyze the average secrecy throughput achieved by each CU in the target cell, by considering the random distribution of users and user scheduling.

When the CUs are randomly distributed in the target cell as a PPP, they are i.i.d uniformly distributed on the disk $b(o. R_c)$ with radial density \cite{DistancesDistribution}
\begin{align}
f_r(x)=\frac{2x}{R_c^2},\textrm{ if }0\leq x\leq R_c.\label{radialdensity}
\end{align}
Then,
the average connection outage probability of each scheduled CU, i.e., $p^{Net}_{co,s}$, is given by
\begin{align}
p^{Net}_{co,s}= \int^{R_c}_0\hat{p}_{co,s}(x)\frac{2x}{R^2}dx.\label{Networkaverageconnectionoutage}
\end{align}

Since TDMA is employed in the cell, each CU has an equal probability to be scheduled for service. The following lemma gives the scheduling probability of a  CU in the target cell.
\begin{lemma}
With PPP-distributed CUs in the target cell, the scheduling probability of a  CU is given by
\begin{align}
\mathbb{P}_{\textrm{U}_{\textrm{s}}}=\frac{1-e^{-\pi R_c^2\lambda_{\textrm{U}}}}{\pi R_c^2\lambda_{\textrm{U}}}.\label{ProbabilitySecureUser}
\end{align}
\end{lemma}
\emph{Proof:}
The proof is given in Appendix E.

We now analyze the secrecy throughput achieved by a  randomly chosen CU in the target cell.
As we mentioned in Section III-B, the upper bound $p^U_{so}$ in (\ref{upperBoundSecrecyOutage}) for the secrecy outage is  tight, which is adopted here to approximate the secrecy outage probability. From (\ref{secrecyThroughput}), the secrecy rate $R_s$ should be calculated from $R_s=R_{t,s}-R_e$. Therefore,  the maximal SINR threshold for satisfying the connection outage constraint, $\beta_{\textrm{B}_{\textrm{s}}}$, and the minimal SIR threshold for satisfying the secrecy outage constraint, $\beta_{\textrm{E}}$, should be obtained under the reliability constraint $p^{Net}_{co,s}\leq \sigma$
and security constraint $p^U_{so}\leq \epsilon$.  Although their analytical results are  difficult to get, from (\ref{Networkaverageconnectionoutage}) and (\ref{upperBoundSecrecyOutage}), we find that their numerical results can be obtained by numerical approaches, e.g., bisection search, since $p^{Net}_{co,s}$, and $p^U_{so}$ are both monotonically  increasing functions of $\beta_{\textrm{B}_{\textrm{s}}}$ and $\beta_{\textrm{E}}$, respectively.

With the obtained numerical results of  $\beta_{\textrm{B}_{\textrm{s}}}$ and $\beta_{\textrm{E}}$, the maximal secrecy rate $R_{t,s}$ and the minimal rate redundancy $R_e$ can be calculated as
$R_{t,s}=\textrm{log}_2\left(1+\beta_{\textrm{B}_{\textrm{s}}}\right)$ and $R_e=\textrm{log}_2(1+\beta_{\textrm{E}})$, respectively. Then, taking  random scheduling into consideration, the maximal secrecy throughput achieved by a randomly chosen CU can be calculated from (\ref{secrecyThroughput}), which is given by
\begin{align}
\mu=\mathbb{P}_{U_S}\left(1-\sigma\right)\left(\textrm{log}_2
\left(1+\beta_{\textrm{B}_{\textrm{s}}}\right)-\textrm{log}_2(1+\beta_{\textrm{E}})\right). \label{ST}
\end{align}

\section{simulation results}
In this section, we provide numerical results to illustrate the effects of different system parameters on the  secrecy performances of the cellular network, i.e.,  the average secrecy throughput of per  CU in the network.

\subsection{Effect of AN power allocation}

\begin{figure}[!t]
\centering
\includegraphics[width=3in]{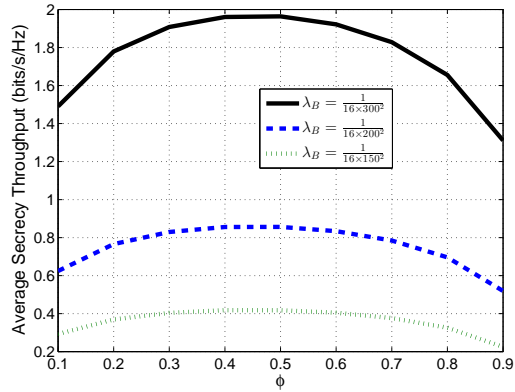}
\caption{The achievable average secrecy throughput versus the power allocation coefficient, $\phi$, between AN and confidential message for different $\lambda_{\textrm{B}}$. The system parameters are $R_c=300$ m, $P_{\tau}=30$dBm, $\lambda_{\textrm{U}}=10\lambda_{\textrm{B}}$, $\tau=N_t=4$, $\lambda_{\textrm{E}}=\lambda_{\textrm{B}}/10$, $N_0=-50$ dBm, the connection outage constraint $\sigma=0.1$, secrecy outage constraint $\epsilon=0.01$, and $P_{\textrm{tot}}=30$dBm.}
\label{PowerAllocationRatio}
\end{figure}

In Fig. \ref{PowerAllocationRatio}, we plot the achievable average secrecy throughput (\ref{ST}) versus the power split factor $\phi$ to show the effect of AN power allocation on the achievable secrecy performance.
As we know, as $\phi$ increases, the power allocated to the AN decreases and the power allocated to the confidential signals increases. With the increasing power of the AN, both the wiretapping capability of eavesdroppers and the reception quality of the intended CU would decrease.  Accordingly, with the increasing power of the confidential signals, both the wiretapping capability of eavesdroppers and the reception quality of the intended CU would increase. Therefore, we can infer that there is an optimal tradeoff between deteriorating the eavesdroppers' wiretapping capability and improving the intended CU's reception quality. This has been validated by the simulation results in Fig. \ref{PowerAllocationRatio}. We can find that the achievable average secrecy throughput first increases and then decreases with the increasing $\phi$, which shows that  AN is helpful for improving the security of the cellular network. But it is important to optimize the power allocation to obtain a preferable secrecy performance.

From above, we can conclude that AN transmission provides a substantial secrecy improvement to  CUs, which is a promising secrecy scheme for a cellular network.

\subsection{Effect of pilot contamination}

\begin{figure}[!t]
\centering
\includegraphics[width=3in]{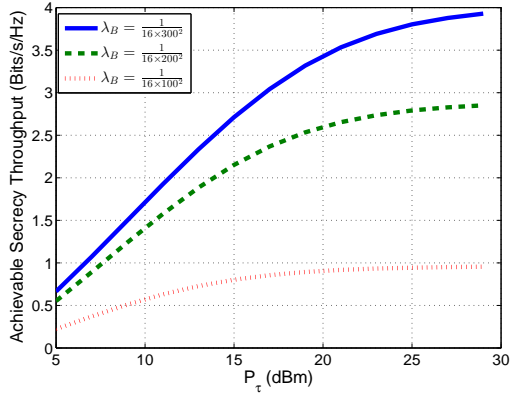}
\caption{The achievable average secrecy throughput versus the pilot power $P_{\tau}$. The system parameters are $R_c=300$ m, $\phi=0.5$ $\alpha=3$, $\lambda_{\textrm{U}}=10\lambda_{\textrm{B}}$, $\tau=N_t=4$, $\lambda_{\textrm{E}}=\lambda_{\textrm{B}}/10$, $N_0=-50$ dBm, the connection outage constraint $\sigma=0.1$, secrecy outage constraint $\epsilon=0.01$, and $P_{\textrm{S}}=15$dBm, $P_{\textrm{A}}=15$ dBm.}
\label{Training}
\end{figure}

In Fig. \ref{Training}, we plot the achievable secrecy throughput versus the pilot power $P_{\tau}$ for different $\lambda_{\textrm{B}}$. From (\ref{estimationerror}), we know that the CSI estimation quality improves with an increasing $P_{\tau}$, which improves the secrecy throughput. Furthermore, the effect of the pilot contamination is also validated in Fig. \ref{Training}. The average secrecy throughput per  CU decreases with an increasing $\lambda_{\textrm{B}}$ since the effect of the pilot contamination  increases as $\lambda_{\textrm{B}}$ increases.
In addition, the secrecy performance gaps between different $\lambda_{\textrm{B}}$'s increase with an increasing $P_{\tau}$. This can be explained by the fact that when $P_{\tau}$ is small, the thermal noise at the target BS dominates the cochannel pilot interference during the uplink training from CUs outside the target cell. Therefore, the performance gaps between different $\lambda_{\textrm{B}}$'s is small. However, when $P_{\tau}$ is large, the pilot contamination interference becomes dominated. When $P_{\tau}$ increases, the interference power causing pilot contamination increases with the increasing $\lambda_{\textrm{B}}$ and
the performance gaps between different $\lambda_{\textrm{B}}$'s get larger.

\subsection{Effect of number of antennas at BSs}

\begin{figure}[!t]
\centering
\includegraphics[width=3in]{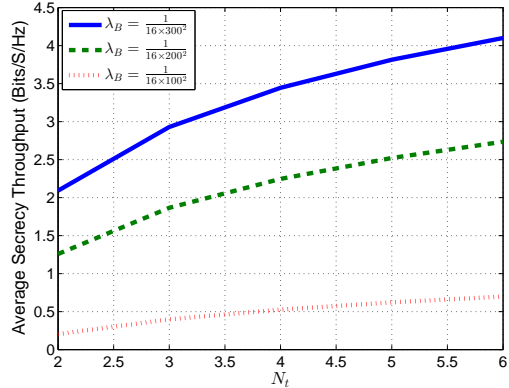}
\caption{The achievable average secrecy throughput versus the number of antennas equipped at each BS for different $\lambda_{\textrm{B}}$'s. The system parameters are $R_c=300$ m, $P_{\tau}=30$dBm, $\lambda_{\textrm{U}_{\textrm{s}}}=10\lambda_{\textrm{B}}$, $\tau=N_t$, $\lambda_{\textrm{E}}=\lambda_{\textrm{B}}/10$, $N_0=-50$ dBm,, the connection outage constraint $\sigma=0.1$, secrecy outage constraint $\epsilon=0.01$, $P_{\textrm{tot}}=30$dBm, and $\phi=0.3$.}
\label{changeNt}
\end{figure}

In Fig. \ref{changeNt}, we plot the  achievable average secrecy throughput versus the number of antennas equipped at each BS ($N_t$) to show the secrecy performance gains brought by an increasing $N_t$. With an increasing $N_t$, the strength of confidential signals would increase due to an increasing diversity gain. Furthermore, AN would interfere with the potential eavesdroppers more efficiently due to an increasing degree-of-freedom.  As expected,  we can find that the achievable average secrecy throughput increases monotonically. It is also shown that the achievable average secrecy throughput decreases with an increasing $\lambda_{\textrm{B}}$. This result can be explained by the fact that the harmful interference received at CUs increases with an increasing $\lambda_{\textrm{B}}$.

\subsection{Secrecy performance of the CU at the cell edge}
\begin{figure}[!t]
\centering
\includegraphics[width=3in]{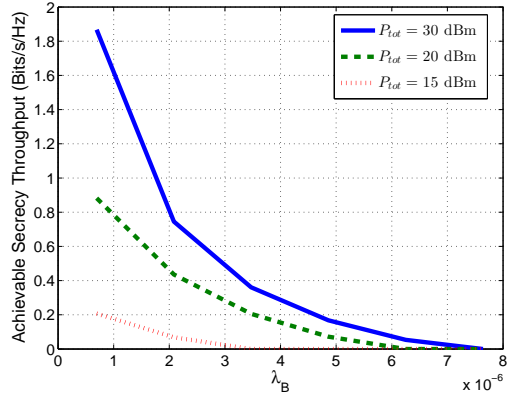}
\caption{The achievable average secrecy throughput of the CU at the edge of the target cell versus the intensity of BSs. The system parameters are $R_c=100$ m, $\alpha=3$, $\lambda_U=10\lambda_{\textrm{B}}$, $\tau=N_t=3$, $P_{\tau}=30$dBm, $N_0=-50$ dBm, $\lambda_{\textrm{E}}=\frac{\lambda_{\textrm{B}}}{10}$, the connection outage constraint $\sigma=0.1$, secrecy outage constraint $\epsilon=0.01$, and $\phi=0.3$.}
\label{CellEdgeUser}
\end{figure}

In the above simulation results, we have investigated the average secrecy performance achieved by a randomly chosen CU in the target cell. But, the secrecy performance of the CU at the cell edge is unknown. As we known, the power of the inter-cell interference received by the CU at the cell edge is larger than any other CUs in the target cell. Therefore, the secrecy performance of the CU at the cell edge represents the worst-case secrecy performance of the target cell, which should be studied separately.

Fig. \ref{CellEdgeUser} gives the simulation results of the secrecy throughput achieved by the CU at the edge of the target cell versus the intensity of BSs.
With the increasing intensity of BSs, the inter-cell interference power received by the CU at the cell edge would increase and the receive performance of the CU would be deteriorated. This has been validated by the simulation results in  Fig. \ref{CellEdgeUser}. From Fig. \ref{CellEdgeUser}, we can find that the achievable secrecy throughput decreases with the increasing $\lambda_B$.
Furthermore, from Fig. \ref{CellEdgeUser}, we can find that the achievable secrecy throughput increases with the increasing $P_{tot}$. This can be explained by the following facts. Although the inter-cell interference power increases with the increasing $P_{tot}$, the power of the confidential signals received at the CU would also increases. When the received SNR at the CU is not large enough, the increasing power of the confidential signals would improve the secrecy performance of the CU.

%
%

\section{Conclusion}
From a network perspective, we analyzed the effect of AN on the secrecy performance of CUs in a randomly deployed cellular network. Based on a hybrid model, we took into account pilot contamination, and derived the analytical average connection outage and the secrecy outage of a  CU. These results facilitate efficient numerical evaluations of the average secrecy throughput. From the numerical results, we find that for a CU, AN is  a promising solution to enhance secrecy in a cellular network. However, there is an optimal tradeoff between jamming eavesdroppers and improving the receiving performance of CUs. Therefore, for maximizing the secrecy performance, the power allocated to AN should be optimized carefully.

\appendices
\section{Proof of Lemma 1}
Since the matrix $\left[\mathbf{w}_{z},\mathbf{U}_{z}\right]$ is unitary
and the elements of $\mathbf{f}_{z}$ are independent complex Gaussian distributed with zero mean and unit variance,
$||\mathbf{f}_{z}^H\mathbf{U}_{z}||^2_F\sim\textrm{Gamma}(N_t-1,1)$ and
$|\mathbf{f}_{z}^H\mathbf{w}_{z}|^2\sim\mathrm{exp}(1)$. Therefore, when $P_{\textrm{S}}=\frac{P_{\textrm{A}}}{N_t-1}$, $P_{z} \sim
\textrm{Gamma}(N_t,P_{\textrm{S}})$.

When $P_{\textrm{S}}\neq \frac{P_{\textrm{A}}}{N_t-1}$, the pdf of $P_{z}$, i.e., $f_{P_{z}}$ can be calculated as
\begin{align}
f_{P_{z}}(x)&=\int^{x}_0\frac{1}{P_{\textrm{S}}}\mathrm{exp}\left(-\frac{x-y}{P_{\textrm{S}}}\right)\frac{y^{N_t-2}e^{-\frac{y(N_t-1)}{P_{\textrm{A}}}}}{\left(\frac{P_{\textrm{A}}}{N_t-1}\right)^{N_t-1}\Gamma(N_t-1)}dy
\nonumber\\
&=\frac{\left(1-\frac{P_{\textrm{A}}}{(N_t-1)P_{\textrm{S}}}\right)^{1-N_t}}{P_{\textrm{S}}\Gamma(N_t-1)}\gamma\left(N_t-1,\left(\frac{N_t-1}{P_{\textrm{A}}}-\frac{1}{P_{\textrm{S}}}\right)x\right)
\mathrm{exp}\left(-\frac{x}{P_{\textrm{S}}}\right).
\end{align}

\section{Proof of Theorem 1}
From (\ref{Connoutage}),
since $\hat{\mathbf{h}}_o\sim \mathcal{CN}\left(\mathbf{0},\delta^2\mathbf{I}_{N_t}\right)$, $||\hat{\mathbf{h}}_o||_F^2\sim\textrm{Gamma}(N_t,\delta^2)$ and its cumulative distribution function (CDF) is given by
\begin{align}
\textrm{Pr}\left(||\hat{\mathbf{h}}_o||_F^2\leq z\right) = 1-\sum_{k=0}^{N_t-1}\left(\frac{z}{\delta^2}\right)^k\frac{1}{k!}\mathrm{exp}\left(-\frac{z}{\delta^2}\right).
\end{align}
Defining $x_{p,s}\triangleq \frac{\left(-1\right)^p\mu_s^p}{p!}\frac{d^p\mathcal{L}_{\hat{I}_{out}(\mu_s)}}{d^p\mu_s}$, we have
\begin{align}
&\hat{p}_{co,s}=1-\sum_{k=1}^{N_t-1}\mathrm{exp}\left(-\mu_s P_{\textrm{I}}\right)\sum^{k}_{p=0}\frac{\left(\mu_sP_{\textrm{I}}\right)^{k-p}x_{p,s}}{(k-p)!}.
\end{align}

In the following, we concentrate on deriving the closed-form result of $x_{p,s}$.
Using the probability generating functional (PGFL) \cite{RandomGraphs}, the Laplace transform $\mathcal{L}_{\hat{I}_{out}(\mu_s)}$ can be derived as
\begin{align}
\mathcal{L}_{\hat{I}_{out}(\mu_s)}&=\mathbb{E}_{\hat{\Phi}_{\textrm{B}}}\left(\prod_{z_s\in\hat{\Phi}_{\textrm{B}}/b(o,R_u)}\mathbb{E}_{P_{z_s}}\mathrm{exp}\left(-\mu_s P_{z_s}D_{z_s}^{-\alpha}\right)\right)
=\mathrm{exp}\left(-2\pi\hat{\lambda}_{\textrm{B}}\int^{+\infty}_{R_u}\left(1-w_s(r)\right)rdr
\right),\label{newLaplace1}
\end{align}
where $w_s(r)\triangleq \mathbb{E}_{P_{z}}\left(\mathrm{exp}\left(-\mu_s P_{z}r^{-\alpha}\right)\right)$.

The pdf of $P_{z}$ has been given in Lemma 1. For brevity, we only consider the case $P_{\textrm{S}}\neq \frac{P_{\textrm{A}}}{N_t-1}$, and the analysis result of $w_s(r)$ for the case $P_{\textrm{S}}= \frac{P_{\textrm{A}}}{N_t-1}$, can be obtained by a similar way.
\begin{align}
w_s(r)=\frac{\left(1-\frac{P_{\textrm{A}}}{(N_t-1)P_{\textrm{S}}}\right)^{1-N_t}}{P_{\textrm{S}}}\left(\frac{P_{\textrm{S}}}{1+P_{\textrm{S}}\mu_s r^{-\alpha}}
-\sum_{i=0}^{N_t-2}\frac{\left(\frac{N_t-1}{P_{\textrm{A}}}-\frac{1}{P_{\textrm{S}}}\right)^i}{\left(\frac{N_t-1}{P_{\textrm{A}}}+\mu_sr^{-\alpha}\right)^{i+1}}
\right).
\end{align}

Then, with (\ref{newLaplace1}), we have
\begin{align}
\frac{d\mathcal{L}_{\hat{I}_{out}}(\mu_s)}{d\mu_s}&=
\left(2\pi\hat{\lambda}_{\textrm{B}}\int^{+\infty}_{R_u}\frac{dw_s(r)}{d_{\mu_s}}rdr+
\right)
\mathcal{L}_{\hat{I}_{out}(\mu_s)}
=\mathcal{L}_{\hat{I}_{out}(\mu_s)}g\left(\mu_s\right)
\end{align}
where
\begin{align}
&g\left(\mu_s\right)\triangleq2\pi\hat{\lambda}_{\textrm{B}}\frac{\left(1-\frac{P_{\textrm{A}}}{(N_t-1)P_{\textrm{S}}}\right)^{1-N_t}}{P_{\textrm{S}}}
\int^{+\infty}_{R_u}\frac{-P_{\textrm{S}}^2r^{-\alpha}}{\left(1+P_{\textrm{S}}\mu_sr^{-\alpha}\right)^2}-\sum_{i=0}^{N_t-2}\frac{\left(\frac{N_t-1}{P_{\textrm{A}}}-\frac{1}{P_{\textrm{S}}}\right)^i(i+1)\left(-r^{-\alpha}\right)}{\left(\frac{N_t-1}{P_{\textrm{A}}}+\mu_sr^{-\alpha}\right)^{i+2}}
rdr.
\end{align}
Then, applying the Leibniz formula, we have
\begin{align}
x_{p,s}&=\frac{\left(-1\right)^p\mu_s^p}{p!}\sum^{p-1}_{m=1}\binom{p-1}{m}\frac{d^{p-1-m} g\left(\mu_s\right)}{d^{p-1-m}\mu_s}\frac{m!}{\left(-1\right)^m\mu_s^m}x_{m,s}
=\sum^{p-1}_{m=1}\frac{p-m}{p}\Psi_{p-m}x_{m,s},\label{newLaplace2}
\end{align}
where
\begin{align}
\Psi_{p-m}\triangleq&\mu_s^{p-m}\left(
2\pi\hat{\lambda}_{\textrm{B}}\frac{\left(1-\frac{P_{\textrm{A}}}{(N_t-1)P_{\textrm{S}}}\right)^{1-N_t}}{P_{\textrm{S}}}\right.
\nonumber\\
&\left.
\underset{I_1}{\underbrace{\int^{+\infty}_{R_u}\left(\frac{P_{\textrm{S}}\left(P_{\textrm{S}}r^{-\alpha}\right)^{p-m}}{\left(1+P_{\textrm{S}}\mu_sr^{-\alpha}\right)^{p-m+1}}-\sum_{i=0}^{N_t-2}\binom{i+p-m}{i}\frac{\left(\frac{N_t-1}{P_{\textrm{A}}}-\frac{1}{P_{\textrm{S}}}\right)^i\left(r^{-\alpha}\right)^{p-m}}{\left(\frac{N_t-1}{P_{\textrm{A}}}+\mu_sr^{-\alpha}\right)^{i+p-m+1}}
\right)rdr}}
\right).
\end{align}
Then, employing \cite[eq. (3.194.1)]{Table}, the integral form $I_1$ can be derived as $I_1=\Upsilon_{p-m}$ in (\ref{Upsilon2m}) .

Since the linear recurrence relation of $x_{p,s}$ in (\ref{newLaplace2}) has a similar form as \cite[eq. (37)]{ChangLi}, we can obtain the explicit form of $x_{p,s}$ with a similar procedure in \cite{ChangLi} which is given by
\begin{align}
x_{p,s}=\sum_{m=1}^{N_t-1}\frac{\mathbf{Q}^m(p+1,1)}{m!}
\mathcal{L}_{\hat{I}_{out}}.
\end{align}
$\mathcal{L}_{\hat{I}_{out}}$ can be derived as follows
\begin{align}
&\mathcal{L}_{\hat{I}_{out}}
\overset{(b)}{=}\mathrm{exp}\left(-2\pi\hat{\lambda}_{\textrm{B}}\left(
\int^{+\infty}_0\left(1-w_s(r)\right)rdr-\int^{R_u}_0\left(1-w_s(r)\right)rdr
\right)\right)
\nonumber\\
&\qquad\overset{(c)}{=}\mathrm{exp}\left(-\hat{\lambda}_{\textrm{B}}\pi\mathbb{E}\left(P_{z_s}^{\delta}\right)\Gamma(1-\delta)\mu_s^{\delta}\right)
\mathrm{exp}\left(\underset{I_2}{\underbrace{2\pi\hat{\lambda}_{\textrm{B}}\int^{R_u}_0\left(1-w_s(r)\right)}}\right)rdr\label{laplacetransformAppendix}
\end{align}
Step $(b)$
follows from the probability generating functional (PGFL) of a PPP. Step $(c)$ is due to \cite[eq. (3.194.2)]{Table} and \cite[eq. (8)]{RandomGraphs}.
%

With the pdf of $P_{z_s}$ in (\ref{PZpdf}) and \cite[eq. (6.455.2)]{Table}, $\mathbb{E}\left(P_{z_s}^{\delta}\right)$ in (\ref{laplacetransformAppendix}) can be derived as
\begin{align}
\mathbb{E}\left(P_{z_s}^{\delta}\right)=\frac{\gamma(\delta+N_t)}{P_{\textrm{S}}\gamma(Nt)\left(\frac{Nt-1}{PA}\right)^{\delta+1}}{_2}F_1\left(1,N_t+\delta;N_t;1-\frac{PA}{(Nt-1)P_{\textrm{S}}}\right)
\end{align}

Using the variable substitution: $z=r^{-\alpha}$, the integral term $I_2$ in (\ref{laplacetransformAppendix}) can be derived as
\begin{align}
I_2&=\pi\hat{\lambda}_{\textrm{B}}R_u^2-\frac{2\pi\hat{\lambda}_{\textrm{B}}\left(1-\frac{P_{\textrm{A}}}{(N_t-1)P_{\textrm{S}}}\right)^{1-N_t}}{\alpha P_{\textrm{S}}}\int^{+\infty}_{R_u^{-\alpha}}\left(\frac{P_{\textrm{S}}z^{-\frac{2}{\alpha}-1}}{1+p_s\mu_sz}-
\sum^{N_t-2}_{i=0}\frac{\left(1-\frac{P_{\textrm{A}}}{P_{\textrm{S}}\left(N_t-1\right)}\right)^i}{\frac{N_t-1}{P_{\textrm{A}}}}\frac{z^{-\frac{2}{\alpha}-1}}{1+\frac{P_{\textrm{A}}\mu_sz}{N_t-1}}
\right)dz
\end{align}
With \cite[eq. (3.194.2)]{Table}, $I_2$ can be further derived as $T\left(\mu_i\right)$ in (\ref{tmui}).

\section{Proof of Theorem 2}
To prove Theorem 2, we need the following lemma,
\begin{lemma}[Alzer's inequality \cite{Alzer}]
If $x\sim \textrm{Gamma}(N,1)$ , then the CDF $F_{x}(y)=\textrm{Pr}\left(x\leq y\right)$ is tightly lower bounded by
$
\left(1-e^{-\kappa y}\right)^N\lessapprox F_{x}(y),
$
where $F_x(y)=\int^{y}_0\frac{e^{-x}x^{N-1}}{\left(N-1\right)!}dx$
and $\kappa=\left(N!\right)^{-\frac{1}{N}}$.
\end{lemma}

Since $||\hat{\mathbf{h}}_B||_F^2\sim \textrm{Gamma}(N_t,\delta^2)$, according to Alzer's inequality \cite{Alzer}, the tight lower bound of $\hat{p}_{co,i}$ for $i=s,c$ are given as follows
\begin{align}
\hat{p}_{co,i}\gtrapprox \left(1-\mathrm{exp}\left(-\kappa\mu_i\left(P_{\textrm{I}}+\hat{I}_{out}\right)\right)\right)^{N_t}.
\end{align}
Using the binomial expansion, the proof can be completed.

\section{Proof of Theorem 3}
In the following proof, we set $\chi_e\triangleq \mathbf{g}_{eo}^H\mathbf{w}_o\mathbf{w}_o^H\mathbf{g}_{eo}$, and
$\omega_{ez}\triangleq \mathbf{g}_{ez}^H\mathbf{U}_z\mathbf{U}_z^H\mathbf{g}_{ez}$, $z\in \Phi_{\textrm{E}}\cup \{o\}$. Since $\mathbf{g}_{ez}$ is independent of $\mathbf{w}_z$ and $\mathbf{U}_z$, we can conclude that $\chi_e\sim \mathrm{exp}(1)$ and $\omega_{ez}\sim \textrm{Gamma}\left(N_t-1,1\right)$.

\subsection{Upper Bound}
We first show the derivation of the upper bound  $p^U_{so}$ as follows
\begin{align}
&p_{so}=1-\mathbb{E}_{\hat{\Phi}_{\textrm{B}}}\left(\mathbb{E}_{\Phi_{\textrm{E}}}\left(\prod_{e\in\Phi_{\textrm{E}}}\textrm{Pr}\left(\textrm{SIR}_{E_e}\leq\beta_{\textrm{E}}|\hat{\Phi}_{\textrm{B}}\right)\right)\right)
\nonumber\\
&=1-\mathbb{E}_{\hat{\Phi}_{\textrm{B}}}\left(\mathrm{exp}\left(-2\pi\lambda_{\textrm{E}}\int^{+\infty}_0\textrm{Pr}\left(\frac{P_{\textrm{S}}\chi_ey^{-\alpha}}{\frac{P_{\textrm{A}}}{N_t-1}\omega_{eo}y^{-\alpha}+
{\sum_{z\in {\hat{\Phi}}_{\textrm{B}}/b(0,R_c)}{\frac{P_{\textrm{A}}}{N_t-1}} \omega_{ez}
D_{ez}^{-{\alpha}}}}\geq\beta_{\textrm{E}}\right)ydy|\hat{\Phi}_{\textrm{B}}\right)\right)
\nonumber\\
&\overset{(f)}{\leq} p^U_{so}\triangleq 1-\mathrm{exp}\left(-2\pi\lambda_{\textrm{E}}\int^{+\infty}_0\mathbb{E}_{\hat{\Phi}_{\textrm{B}}}\left(\textrm{Pr}\left(\textrm{SIR}_{E_e}\geq\beta_{\textrm{E}}|\hat{\Phi}_{\textrm{B}}\right)\right)ydy\right)\nonumber\\
&\overset{(g)}{=}
1-\mathrm{exp}\left(-2\pi\lambda_{\textrm{E}}\left(1+\alpha_{\textrm{E}}\right)^{-N_t+1}\int^{+\infty}_0
\mathbb{E}_{\hat{\Phi}_{\textrm{B}}}\left(\mathrm{exp}\left(\sum_{z\in\hat{\Phi}_{\textrm{B}}/b(o,R_c)}f(\omega_{ez}, y, D_{ez})\right)\right)ydy\right).
\end{align}
where $\alpha_{\textrm{E}}\triangleq \frac{P_{\textrm{A}}\beta_{\textrm{E}}}{(N_t-1)P_{\textrm{S}}}$,
$f(\omega_{ez}, y, D_{ez}) \triangleq -\alpha_{\textrm{E}}\omega_{ez}y^{\alpha}D_{ez}^{-\alpha}$, step $(f)$ is due to Jensen's inequality, and step $(g)$ is due to the Laplace transform of the gamma variable.

The difficulty of further derivation lies in  the integral of
$\int^{+\infty}_0
\mathbb{E}_{\hat{\Phi}_{\textrm{B}}}\left(\mathrm{exp}\left(\sum_{z}f(\omega_{ez}, y, D_{ez})\right)\right)ydy$. This is because when $y<R_c$ (eavesdroppers in the target cell) and $y>R_c$ (eavesdroppers outside the target cell), the function $\mathbb{E}_{\hat{\Phi}_{\textrm{B}}}\left(\mathrm{exp}\left(\sum_{z}f(\omega_{ez}, y, D_{ez})\right)\right)$ has different expressions due to the different shapes of the interference region from $\hat\Phi_{\textrm{B}}$, i.e., we have
 \begin{align}
&\int^{+\infty}_0
\mathbb{E}_{\hat{\Phi}_{\textrm{B}}}\left[\mathrm{exp}\left(\sum_{z}f(\omega_{ez}, y, D_{ez})\right)\right]ydy = \nonumber\\
&\int^{R_c}_0
\underset{T_1}{\underbrace{\mathbb{E}\left[\mathrm{exp}\left(\sum_{z}f(\omega_{ez}, y, D_{ez})\right)|y\leq R_c\right]}}ydy
+\int^{+\infty}_{R_c}
\underset{T_2}{\underbrace{\mathbb{E}\left[\mathrm{exp}\left(\sum_{z}f(\omega_{ez}, y, D_{ez})\right)|y>R_c\right]}}ydy
.\label{mediaResult}
\end{align}

In Fig. \ref{Illustration}, we show these two cases.
In the following, we derive the analytical results of $T_1$ and $T_2$.

1. The analytical result of $T_1$.

Fig. \ref{Illustration} (a) shows the case that eavesdroppers are in the target cell, where we have $l_1(\theta)=\sqrt{R_c^2-y^2\textrm{sin}^2\theta}+y\textrm{cos}\theta$ and $l_2(\theta)=\sqrt{R_c^2-y^2\textrm{sin}^2\theta}-y\textrm{cos}\theta$. Then, we have
\begin{align}
&T_1= \mathrm{exp}\left(-\lambda_{\textrm{B}}\left(
\int^{\pi}_0\mathbb{E}_{\omega_{ez}}\left[\int^{+\infty}_{l_2(\theta)}\left(1-\mathrm{exp}\left(f(\omega_{ez}, y, x)\right)\right)xdx\right.
\right.\right.\nonumber\\
&\qquad \qquad \qquad \qquad \ +\left.\left.\left. \int^{+\infty}_{l_1(\theta)}\left(1-\mathrm{exp}\left(f(\omega_{ez}, y, x)\right)\right)xdx\right]d\theta\right)\right).
\end{align}

Then invoking \cite[eq. (28)]{ShotNoiseModel}, we have
\begin{align}
&\mathbb{E}_{\omega_{ez}}\left[\int^{+\infty}_{l_2(\theta)}\left(1-\mathrm{exp}\left(f(\omega_{ez}, y, x)\right)\right)xdx\right] \nonumber \\
=&-\frac{l^2_2(\theta)}{2}\mathbb{E}_{\omega_{ez}}\left[1-\mathrm{exp}\left(f(\omega_{ez}, y, l_2(\theta))\right)\right]
+\frac{1}{2}\mathbb{E}_{\omega_{ez}}\left[\left(-f(\omega_{ez}, y,1)\right)^{\frac{2}{\alpha}}\Gamma\left(1-\frac{2}{\alpha}\right)\right]\nonumber\\
&\qquad -\frac{1}{2}\mathbb{E}_{\omega_{ez}}\left[\left(-f(\omega_{ez}, y,1)\right)^{\frac{2}{\alpha}}\Gamma\left(1-\frac{2}{\alpha}, -f(\omega_{ez}, y, l_2(\theta))\right)\right]
\nonumber\\
\overset{(h)}{=}&\Omega(l_2(\theta))\label{GuardZero}
\end{align}
where step ($h$) can be achieved by adopting  \cite[eq. (3.326.2)]{Table} and \cite[eq. (6.455.1)]{Table}, since $\omega_{ez}\sim\textrm{Gamma}\left(N_t-1,1\right)$.

The analytical result of $\mathbb{E}_{\omega_{ez}}\left(\int^{+\infty}_{l_1(\theta)}\left(1-\mathrm{exp}\left(f(\omega_{ez}, y, x)\right)\right)xdx\right)$ can be obtained with the same procedures, which are omitted for brevity. Then, the analytical result of $T_1$ can be obtained.

\begin{figure}[!t]
\centering
\includegraphics[width=4in]{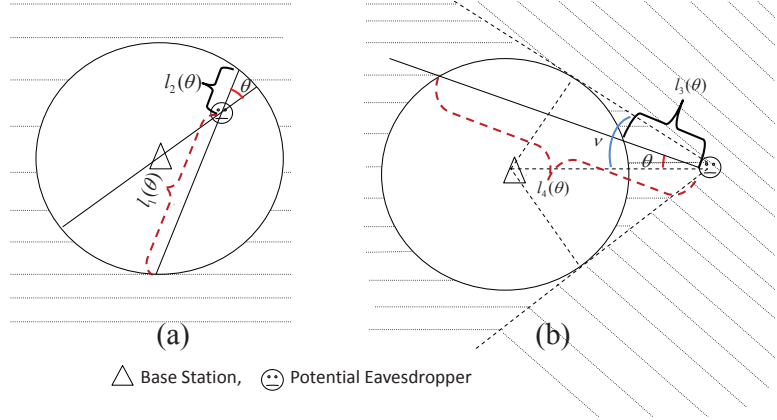}
\caption{Illustration of the regions of interfering BSs: a) eavesdroppers are in the target cell, where the interfering BSs region is labeled by transverse lines; b) eavesdroppers are outside the target cell, where the interfering BSs region has two parts with one labeled by transverse lines and the other labeled by oblique lines. }
\label{Illustration}
\end{figure}

2. The analytical result of $T_2$.

Fig. \ref{Illustration} (b) shows the case that eavesdroppers are outside the target cell, where we have
$\nu=\textrm{arsin}\left(\frac{R_c}{y}\right)$,
$l_3(\theta)=y\textrm{cos}\theta-\sqrt{R_c^2-\left(y\textrm{sin}\theta\right)^2}$, and
$l_4(\theta)=l_3(\theta)+2\sqrt{R_c^2-\left(y\textrm{sin}\theta\right)^2}$.
The interference region could be divided into two parts, as shown in Fig. \ref{Illustration} (b) by different type of lines. Accordingly, $T_2$ can be calculated as
\begin{align}
&T_2=
\nonumber\\
&\mathrm{exp}\left(-2\lambda_{\textrm{B}}\left(\int^{\nu}_0\left(\Xi_1(\theta)+\mathbb{E}_{\omega_{ez}}
\left[\int^{+\infty}_{l_4(\theta)}\left(1-\mathrm{exp}\left(f(\omega_{ez}, y, x)\right)\right)xdx
\right]\right)d\theta
+\int^{\pi}_\nu \Xi_2(\theta)d\theta\right)\right)\label{T2}
\end{align}
where
\begin{align}
&\Xi_1(\theta)\triangleq\mathbb{E}_{\omega_{ez}}\left[\int^{l_3(\theta)}_{0}\left(1-\mathrm{exp}\left(f(\omega_{ez}, y, x)\right)\right)xdx
\right],
\nonumber\\
&\Xi_2(\theta)\triangleq\mathbb{E}_{\omega_{ez}}\left[\int^{+\infty}_{0}\left(1-\mathrm{exp}\left(f(\omega_{ez}, y, x)\right)\right)xdx\right],
\end{align}

Just as (\ref{GuardZero}), we have
\begin{align}
\mathbb{E}_{\omega_{ez}}\left(\int^{+\infty}_{l_4(\theta)}\left(1-\mathrm{exp}\left(f(\omega_{ez}, y, x)\right)\right)xdx
\right)=\Omega(l_4(\theta)).
\end{align}
Invoking \cite[eq. (28)]{ShotNoiseModel}, the analytical results of $\Xi_i(\theta)$, $i=1,2$ can be derived as (\ref{Xi1}) and (\ref{Xi3}), and the details are omitted for brevity. Then substituting $\Xi_1(\theta)$ and $\Xi_3(\theta)$ into (\ref{T2}), the analytical result of $T_2$ can be obtained.

Finally, substituting the analytical result of $T_1$ and $T_2$ into (\ref{mediaResult}), the proof can be completed.

\subsection{Lower Bound}

By considering the nearest eavesdropper only, a lower bound of secrecy outage probability can be derived.
Assuming that the  eavesdropper at $e^*$ is the nearest eavesdropper, $D_{E_{e^*}}$ is distributed according to the
following pdf \cite{DistancesDistribution}:
\begin{align}
f_{D_{E_{e^*}}}(y)=2\pi\lambda_{\textrm{E}}ye^{-\pi\lambda_{\textrm{E}} y^2}.\label{distancepdf}
\end{align}
The lower bound $p^L_{so}$ can be derived as
\begin{align}
&p_{so}\geq p_{so}^L=\mathbb{E}_{D_{E_{e^*}}}\left(\mathbb{E}_{\hat{\Phi}_{\textrm{B}}}\left(\textrm{Pr}\left(\textrm{SIR}_{E_e}\geq z|\hat{\Phi}_{\textrm{B}},D_{E_{e^*}}\right)\right)\right)=
\nonumber\\
&\left(1+\alpha_{\textrm{E}}\right)^{-N_t+1}\int^{+\infty}_0 2\pi\lambda_{\textrm{E}}ye^{-\pi\lambda_{\textrm{E}} y^2}\mathbb{E}\left(\mathrm{exp}\left(\sum_{z\in\hat{\Phi}_{\textrm{B}}/b(o,R_c)} f(\omega_{ez},y,D_{ez})\right)\right)dy.
\end{align}
Therefore, just as the derivation of the derivation of $p_{so}^U$, the key step for getting the analysis result of $p_{so}^L$ is getting the analysis result of $\mathbb{E}\left(\mathrm{exp}\left(\sum_{z\in\hat{\Phi}_{\textrm{B}}/b(o,R_c)}f(\omega_{ez},y,D_{ez})\right)\right)$.
Therefore, following the derivation of $T_1$ and $T_2$, the analytical result of $p_{so}^L$ can be obtained and the detailed derivations are omitted for brevity.

\section{Proof of Lemma 2}
Assuming that the  target CU locates in the target cell in addition to the CUs PPP $\Phi_{\textrm{U}}$. According to the Slivnyak's theorem \cite{RandomGraphs}, we know that the added  target CU does not affect the spatial distribution of other  CUs. Thus, the probability mass function of the number of other  CUs (denoted as $M_s$) is Possion distributed, i.e., $\textrm{Pr}\left(M_s=m\right)=\frac{\left(\pi R_c^2\lambda_{\textrm{U}}\right)^m}{m!}e^{-\pi R_c^2\lambda_{\textrm{U}}}$. Then, the scheduling probability of a  target CU can be evaluated as
\begin{align}
\mathbb{P}_{\textrm{U}}=\sum^{+\infty}_{m=0}\frac{\textrm{Pr}\left(M_s=m\right)}{m+1}=\frac{1-e^{-\pi R_c^2\lambda_{\textrm{U}}}}{\pi R_c^2\lambda_{\textrm{U}}}
\end{align}


\end{document}